\documentclass[letterpaper]{article}
\usepackage{proceed2e}
\usepackage[margin=1in]{geometry}
\usepackage{subfig}
\usepackage{subfiles}

\usepackage{natbib}
\usepackage{amsthm}
\usepackage{amsmath}
\usepackage{amssymb}

\usepackage[algoruled,linesnumbered, onelanguage]{algorithm2e}

\usepackage{enumitem}

\usepackage{mathtools}
\usepackage{pdfpages}

\newtheoremstyle{personal}
{5pt}
{5pt}
{}
{8pt}
{\scshape}
{:}
{.5em}
{}

\newtheorem{definition}{Definition}[section]
\newtheorem{lemma}[definition]{Lemma}
\newtheorem{theorem}[definition]{Theorem}

\newtheorem{example}[definition]{Example}

\newcommand*{\ci}{\!\perp\!\!\!\perp}
\newcommand*{\ncid}{\not\perp}
\newcommand*{\cid}{\perp}

\DeclareMathOperator*{\Pa}{Pa}

\DeclareMathOperator*{\CPDAG}{CPDAG}
\DeclareMathOperator*{\Skeleton}{Skeleton}
\DeclareMathOperator*{\DirectedPart}{DirPart}
\DeclareMathOperator*{\GES}{GES}
\DeclareMathOperator*{\AGES}{AGES}
\DeclareMathOperator*{\AggregateCPDAGs}{AggregateCPDAGs}

\DeclareMathOperator*{\rank}{rank}

\usepackage{tikz}
\usetikzlibrary{arrows,snakes,backgrounds,decorations.pathreplacing}
\tikzstyle{place}=[circle,minimum width=15pt,draw=blue!50,fill=blue!20,thick]
\tikzstyle{transition}=[rectangle,minimum width=15pt,draw=black!50,fill=black!20,thick]
\tikzstyle{node}=[scale=.85,minimum width=15pt,draw=white,fill=white,thick]
\tikzstyle{snode}=[scale=.85,circle,minimum width=32pt,draw=black,fill=white,thick]
\tikzstyle{white}=[scale=.85,circle,minimum width=25pt,draw=white,fill=white,thick]

\usepackage{times}
\usepackage{xr}

\title{Structure Learning of Linear Gaussian Structural Equation Models with Weak Edges}

\author{ {\bf Marco F. Eigenmann} \\
Seminar f\"ur Statistik\\
ETH Zurich\\
Zurich, Switzerland
\And
{\bf Preetam Nandy}  \\
Department of Biostatistics and Epidemiology        \\
University of Pennsylvania \\
Philadelphia, PA 19104, USA \\
\And
{\bf Marloes H. Maathuis}   \\
Seminar f\"ur Statistik \\
ETH Zurich    \\
Zurich, Switzerland\\
}

\begin{document}

\maketitle

\begin{abstract}
   We consider structure learning of linear Gaussian structural equation models with weak edges. Since the presence of weak edges can lead to a loss of edge orientations in the true underlying CPDAG, we define a new graphical object that can contain more edge orientations. We show that this object can be recovered from observational data under a type of strong faithfulness assumption. We present a new algorithm for this purpose, called aggregated greedy equivalence search (AGES), that aggregates the solution path of the greedy equivalence search (GES) algorithm for varying values of the penalty parameter. We prove consistency of AGES and demonstrate its performance in a simulation study and on single cell data from \cite{Sachs}. The algorithm will be made available in the R-package \texttt{pcalg}.
\end{abstract}

\section{INTRODUCTION}

We consider structure learning of linear Gaussian structural equation models (SEMs) \citep{Bollen}. A linear SEM is a set of equations of the form $X = B^T X + \varepsilon$, where $X = (X_1,\dots,X_p)^T$, $B$ is a $p\times p$ strictly upper triangular matrix, $\varepsilon= (\varepsilon_1,\dots,\varepsilon_p)^T$, and $\varepsilon$ is multivariate Gaussian with mean vector zero and a diagonal covariance matrix $D$ (hence assuming no hidden confounders). Such SEMs can be represented by a directed acyclic graph (DAG) $G$, where a nonzero entry $B_{ij}$ corresponds to an edge from $X_i$ to $X_j$. By putting the coefficients $B_{ij}$ along the corresponding edges, one obtains a weighted graph. This weighted graph and the distribution of $\varepsilon$ fully determine the distribution of $X$.
Example~\ref{Example: Motivation} shows a simple instance with $p=3$, where
$$B = \begin{pmatrix}
   0 & 0.1& 1 \\
   0 & 0  & 1\\
   0 & 0  & 0
\end{pmatrix}.$$ The weighted DAG is shown in Figure~\ref{Subfigure Motivating example DAG}.

Based on $n$ i.i.d.\ observations from $X$, we aim to learn the underlying DAG $G$. However, since $G$ is generally not identifiable from the distribution of $X$, we learn the so-called Markov equivalence class of $G$, which can be represented by a completed partially directed acyclic graph (CPDAG) (see Section~\ref{Subsection Terminology graphical models}). A CPDAG can contain both directed and undirected edges, where undirected edges represent uncertainty about the edge orientation.

Several efficient algorithms have been developed to learn CPDAGs, such as for example the PC algorithm \citep{Spirtes} and the greedy equivalence search algorithm (GES) \citep{ChickeringGES}. These algorithms have been proved to be sound and consistent \citep{Spirtes, Kalisch, ChickeringGES, Nandy}.

Example~\ref{Example: Motivation} illustrates a somewhat counter-intuitive behaviour of these algorithms for varying sample size.

\begin{example}\label{Example: Motivation}
   Consider the following SEM:
   \begin{align*}
   X_1 &= \varepsilon_{1}\\
   X_2 &= 0.1\cdot X_1 + \varepsilon_{2}\\
   X_3 &= X_1 + X_2 + \varepsilon_{3},
   \end{align*}
   where $\varepsilon \sim N(0,I)$.
   The corresponding CPDAG is the complete undirected graph in Figure~\ref{Subfigure Motivating example true CPDAG}.
   When running PC or GES with a very large sample size, the algorithms will output this CPDAG with high probability. For a smaller sample size, however, the algorithms are likely to miss the weak edge $X_{1}-X_2$, leading to the CPDAG in Figure~\ref{Subfigure Motivating example smaller CPDAG}. Note that the latter CPDAG contains two edge orientations that are identical to the orientations in the underlying DAG $G_0$. Thus, both CPDAGs in Figures~\ref{Subfigure Motivating example true CPDAG} and \ref{Subfigure Motivating example smaller CPDAG} contain some relevant information, one in terms of correct adjacencies and one in terms of correct edge orientations. For this example, GES outputs Figure~\ref{Subfigure Motivating example smaller CPDAG} for a sample size smaller than $100$ and Figure~\ref{Subfigure Motivating example true CPDAG} for a sample size larger than 1000, with high probabilities.
\end{example}

\begin{figure}[t]
   \centering
   \hspace{0.7cm}
   \subfloat[True DAG $G_0$.]{
   \label{Subfigure Motivating example DAG}
   \begin{tikzpicture}[scale=0.7]
      \node[node] (X) at (-1.4,1) {$X_1$};
      \node[node] (Y) at (1.4,1) {$X_2$};
      \node[node] (Z) at (0,-1) {$X_3$};
      \draw[thick,->] (X) to [above] node{$0.1$} (Y);
      \draw[thick,->] (X) to [left] node{$1$} (Z);
      \draw[thick,->] (Y) to [right] node{$1$} (Z);
   \end{tikzpicture}
   }
   \hfill
   \subfloat[CPDAG of $G_0$.]{
   \label{Subfigure Motivating example true CPDAG}
   \begin{tikzpicture}[scale=0.7]
      \node[node] (X) at (-1.4,1) {$X_1$};
      \node[node] (Y) at (1.4,1) {$X_2$};
      \node[node] (Z) at (0,-1) {$X_3$};
      \draw[-] (X) to [above]  (Y);
      \draw[-] (X) to [left]  (Z);
      \draw[-] (Y) to [right]  (Z);
   \end{tikzpicture}
   }
   \hspace{0.7cm}
   \vfill
   \hspace{0.7cm}
   \subfloat[CPDAG without the weak edge.]{
   \label{Subfigure Motivating example smaller CPDAG}
   \begin{tikzpicture}[scale=0.7]
      \node[node] (X) at (-1.4,1) {$X_1$};
      \node[node] (Y) at (1.4,1) {$X_2$};
      \node[node] (Z) at (0,-1) {$X_3$};
      \draw[thick,->] (X) to [left]  (Z);
      \draw[thick,->] (Y) to [right]  (Z);
   \end{tikzpicture}
   }
   \hfill
   \subfloat[Desired APDAG.]{
   \label{Subfigure Motivating example APDAG}
   \begin{tikzpicture}[scale=0.7]
      \node[node] (X) at (-1.4,1) {$X_1$};
      \node[node] (Y) at (1.4,1) {$X_2$};
      \node[node] (Z) at (0,-1) {$X_3$};
      \draw[-] (X) to [above]  (Y);
      \draw[thick,->] (X) to [left]  (Z);
      \draw[thick,->] (Y) to [right] (Z);
   \end{tikzpicture}
   }
   \hspace{0.7cm}
   \hfill
   \caption{A simple case where the inclusion of a weak edge leads to a loss of edge orientations (see Example~\ref{Example: Motivation}).}
   \label{Figure Motivating Example1}
\end{figure}
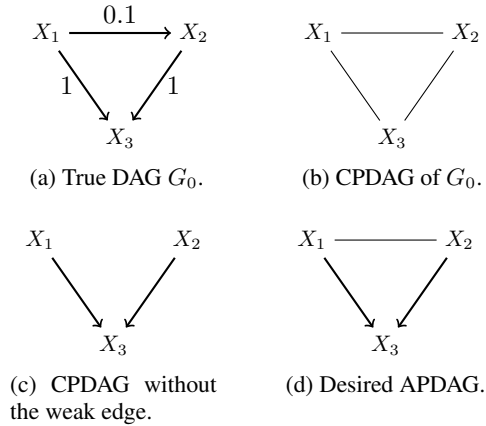

One may think that a simple solution to the above problem is to omit weak edges either by using a strong penalty on model complexity or by truncating edges with small weights. In some cases, however, the inclusion of weak edges can also help to obtain edge orientations. This is illustrated in Example~\ref{Example: Weak helps}.

\begin{example}\label{Example: Weak helps}
    Consider the weighted DAG in Figure~\ref{Subfigure Weak helps weighted DAG} with $\varepsilon\sim N(0,I)$.
    Figure~\ref{Subfigure Weak helps DAG} represents the corresponding CPDAG, which is fully oriented. For large sample sizes, PC and GES will output this CPDAG with high probability. For smaller sample sizes, however, they are likely to miss the weak edge $X_4 \rightarrow X_2$, leading to the CPDAG in Figure~\ref{Subfigure Weak helps smaller CPDAG}, which is fully undirected.
\end{example}

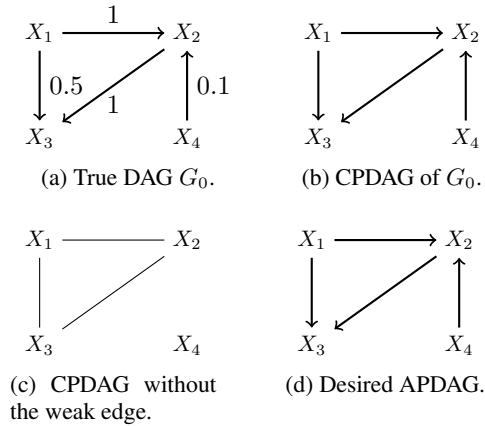
\begin{figure}[t]
    \centering
    \hspace{0.7cm}
    \subfloat[True DAG $G_0$.]{
    \label{Subfigure Weak helps weighted DAG}
    \begin{tikzpicture}[scale=0.7]
       \node[node] (X) at (-1.4,1) {$X_1$};
       \node[node] (Y) at (1.4,1) {$X_2$};
       \node[node] (Z) at (-1.4,-1) {$X_3$};
       \node[node] (W) at (1.4,-1) {$X_4$};
       \draw[thick,->] (X) to [above] node{$1$} (Y);
       \draw[thick,->] (X) to [right] node{$0.5$} (Z);
       \draw[thick,->] (Y) to [below] node{$1$} (Z);
       \draw[thick,->] (W) to [right] node{$0.1$} (Y);
    \end{tikzpicture}
    }
    \hfill
    \subfloat[CPDAG of $G_0$.]{
    \label{Subfigure Weak helps DAG}
    \begin{tikzpicture}[scale=0.7]
       \node[node] (X) at (-1.4,1) {$X_1$};
       \node[node] (Y) at (1.4,1) {$X_2$};
       \node[node] (Z) at (-1.4,-1) {$X_3$};
       \node[node] (W) at (1.4,-1) {$X_4$};
       \draw[thick,->] (X) to [above]  (Y);
       \draw[thick,->] (X) to [left]  (Z);
       \draw[thick,->] (Y) to [right]  (Z);
       \draw[thick,->] (W) to [right] (Y);
    \end{tikzpicture}
    }
    \hspace{0.7cm}
    \vfill
    \hspace{0.7cm}
    \subfloat[CPDAG without the weak edge.]{
    \label{Subfigure Weak helps smaller CPDAG}
    \begin{tikzpicture}[scale=0.7]
       \node[node] (X) at (-1.4,1) {$X_1$};
       \node[node] (Y) at (1.4,1) {$X_2$};
       \node[node] (Z) at (-1.4,-1) {$X_3$};
       \node[node] (W) at (1.4,-1) {$X_4$};
       \draw[-] (X) to [above]  (Y);
       \draw[-] (X) to [left]  (Z);
       \draw[-] (Y) to [right] (Z);
    \end{tikzpicture}
    }
    \hfill
    \subfloat[Desired APDAG.]{
    \label{Subfigure Weak helps APDAG}
    \begin{tikzpicture}[scale=0.7]
       \node[node] (X) at (-1.4,1) {$X_1$};
       \node[node] (Y) at (1.4,1) {$X_2$};
       \node[node] (Z) at (-1.4,-1) {$X_3$};
       \node[node] (W) at (1.4,-1) {$X_4$};
       \draw[thick,->] (X) to [above] (Y);
       \draw[thick,->] (X) to [left] (Z);
       \draw[thick,->] (Y) to [right]  (Z);
       \draw[thick,->] (W) to [right] (Y);
    \end{tikzpicture}
    }
    \hspace{0.7cm}
    \hfill
    \caption{A simple case where the inclusion of a weak edge helps to obtain edge orientations (see Example~\ref{Example: Weak helps}).}
    \label{Figure Motivating Example2}
\end{figure}

With a larger sample size we expect to gain more insight into a system, and the fact that we can lose correct edge orientations is undesirable. In the extreme case of a complete DAG with many weak edges, a small sample size yields informative output in terms of certain edge orientations, while a large sample size yields the asymptotically correct CPDAG, which is the uninformative complete undirected graph. This problem is relevant in practice for situations where the underlying system contains many weak effects and the sample size can be very large.

We propose a solution for this problem by defining a new graphical target object that can contain more edge orientations than the CPDAG. This object is a partially directed acyclic graph (PDAG) obtained by aggregating several CPDAGs of sub-DAGs of the underlying DAG $G_{0}$, and is called an aggregated PDAG (APDAG). In Example~\ref{Example: Motivation}, we intuitively overlay the CPDAGs of Figures~\ref{Subfigure Motivating example true CPDAG} and \ref{Subfigure Motivating example smaller CPDAG} to obtain the APDAG in Figure~\ref{Subfigure Motivating example APDAG}, that contains both the correct skeleton and some edge orientations that were present in $G_0$ but not in the CPDAG of $G_0$. The APDAG for Example~\ref{Example: Weak helps} is given in Figure~\ref{Subfigure Weak helps APDAG}, and is in this case identical to the CPDAG of $G_0$.

Our APDAG is a maximally oriented PDAG, as studied in \cite{Meek}.
We will show that APDAGs can be learned from observational data under a type of strong faithfulness condition, namely strong faithfulness with respect to a sequence of sub-DAGs of the underlying DAG. In this sense, our work is related to other structure learning algorithms that output maximally oriented PDAGs or DAGs  under certain restrictions on the model class \citep[e.g.,][]{ShimizuEtAl06-JMLR, hoyer08, petersbuehlmann14, ernestroth2016, Peters2014}. \cite{Ema2017} provide methods for causal reasoning with maximally oriented PDAGs.

We propose an algorithm to learn APDAGs, by aggregating the solution path of the greedy equivalence search (GES) algorithm for varying values of the penalty parameter. The algorithm is therefore called aggregated GES (AGES).
We show that the entire solution path can basically be computed at once, similarly to the computation of the solution path of the Lasso \citep{Lasso, LassoSolutionPath}. We also prove consistency of the algorithm, and demonstrate its performance in a simulation study and on data from \cite{Sachs}. All proofs are given in the supplementary material.


\section{PRELIMINARIES}\label{Section: Preliminaries}

\subsection{GRAPHICAL MODELS}\label{Subsection Terminology graphical models}

We now introduce the main terminology for graphical models that we need. Further definitions can be found in Section~\ref{Section: Preliminaries SuppMat} of the supplementary material.

A \emph{graph} $G=(X,E)$ consists of a set of \emph{vertices} $X=\{ X_{1}, \ldots ,X_{p} \}$ and a set of \emph{edges} $E$. The edges can be either \emph{directed} $X_{i} \rightarrow X_{j}$ or \emph{undirected} $X_{i} - X_{j}.$ A \emph{directed graph} is a graph that contains only directed edges. A \emph{partially directed graph} can contain both directed and undirected edges.

If $X_{i}\rightarrow X_{j}$, then $X_{i}$ is a \emph{parent} of $X_{j}$. The set of parents of $X_{i}$ in a graph $G$ is denoted by $\Pa_{G}(X_{i})$.
A triple $(X_{i},X_{j},X_{k})$ in a graph $G$ is called a \emph{v-structure} if $X_{i} \rightarrow X_{j} \leftarrow X_{k}$ and $X_{i}$ and $X_{k}$ are not adjacent in $G.$

A \emph{directed acyclic graph} (DAG) is a directed graph that does not contain directed cycles. A partially directed graph that does not contain directed cycles is a \emph{partially directed acyclic graph} (PDAG). A PDAG $P$ is \emph{extendible} to a DAG if the undirected edges of $P$ can be oriented to obtain a DAG without additional v-structures. The \emph{skeleton} of a partially directed graph $G$ is the graph obtained by replacing all directed edges by undirected edges, and is denoted by $\Skeleton(G).$ The \emph{directed part} of a partially directed graph $G$ is the graph obtained by removing all undirected edges, and is denoted by $\DirectedPart(G).$ A DAG $G$ restricted to a graph $H$ is the DAG $G^{\prime}$ obtained by removing from $G$ all adjacencies not present in $H.$ A DAG $G^{\prime}=(X,E^{\prime})$ is a \emph{sub-DAG} of a DAG $G=(X,E)$ if $E^{\prime}\subseteq E.$

A DAG encodes conditional independence constraints via the concept of d-separation \citep{Pearl2009}. Several DAGs can encode the same set of d-separations. Such DAGs are called \emph{Markov equivalent}. Markov equivalent DAGs have the same skeleton and the same v-structures \citep{VermaPearl}. A Markov equivalence class of DAGs can be represented by a \emph{completed partially directed acyclic graph} (CPDAG) \citep{Andersson, Chickering2002}. We denote by $\CPDAG(G)$ the CPDAG of a DAG $G$. A directed edge $X_i\to X_j$ in a CPDAG means that $X_i\to X_j$ occurs in all DAGs in the Markov equivalence class. An undirected edge $X_i-X_j$ in a CPDAG means that there is a DAG with $X_i\to X_j$ and a DAG with $X_i \leftarrow X_j$ in the Markov equivalence class.

We denote conditional independence of two variables $X_{i}$ and $X_{j}$ given a set $S\subseteq X\setminus \{X_{i},X_{j}\}$ by $X_{i} \ci X_{j} \vert S$, and the corresponding d-separation relation in a DAG $G$ is denoted by $X_{i} \cid_G X_{j} \vert S$.

A DAG $G=(X,E)$ is a perfect map of the distribution of $X$ if every conditional independence constraint in the distribution is also encoded by the DAG $G$ via d-separation, and vice versa. The first direction is known as the \emph{faithfulness condition} while the backward direction is known as the \emph{Markov condition}.
A multivariate Gaussian distribution is said to be \emph{$\delta$-strong faithful} to a DAG $G=(X,E)$ if for every $X_{i},X_{j}\in X$ and for every $S\subseteq X\setminus \{X_{i},X_{j}\}$ it holds that $X_{i} \ncid_{G} X_{j} \vert S \Rightarrow \vert \rho_{X_{i},X_{j}\vert S} \vert > \delta,$ where $\rho_{X_{i},X_{j} \vert S}$ is the partial correlation between $X_i$ and $X_j$ given $S$  \citep[cf.,][]{Zhang}.
Faithfulness is a special case of $\delta$-strong faithfulness with $\delta=0$.

Throughout the paper we consider distributions of $X$ that allow a perfect map representation through a DAG $G_{0}=(X,E)$. The density $f$ of $X$ then admits the following factorization based on $G_0$: $f(x)=\prod_{i=1}^{p} f(x_{i}\vert \Pa_{G_{0}}(x_{i})).$

We denote $n$ i.i.d.\ observations of $\tilde{X}\subseteq X$ by $\tilde{X}^{(n)}$. DAGs will be denoted with the letter $G,$ PDAGs with $P,$ CPDAGs with $C,$ and APDAGs with $A.$ We reserve the subscript $0$ for graphs associated with the true underlying distribution.

\subsection{STRUCTURE LEARNING ALGORITHMS}
\label{Subsection: TERMINOLOGY FOR STRUCTURE LEARNING ALGORITHMS}

We will make use of the Greedy Equivalence Search (GES) algorithm of \cite{ChickeringGES}. This algorithm is composed of two phases called the forward and the backward phase. Starting generally from the empty graph, the forward phase greedily adds edges, one at a time, minimizing each time a scoring criterion over the set of neighbouring CPDAGs. The forward phase stops when the score can no longer be improved by a single edge addition. At that point, the backward phase starts and removes edges, also one at a time, minimizing each time the same scoring criterion, until the score can no longer be improved.

GES operates on the space of CPDAGs. Conceptually, a move from one CPDAG to the next goes as follows: GES computes all DAGs belonging to the actual CPDAG. It then computes all possible edge additions (deletions) for each of the found DAGs. Among all possible edge additions (deletions) it chooses the one that leads to the maximum score improvement, and then computes the CPDAG of the resulting DAG. \cite{ChickeringGES} presented an efficient way to move from one CPDAG to the next without computing the DAGs as described above.

GES has one tuning parameter which we call penalty parameter and denote by $\lambda$.
As scoring criterion we take a penalized negative log-likelihood function of the following form:
\begin{align*}
&\mathcal{S}_{\lambda}(G,X^{(n)}) \\
&= -\sum_{i=1}^{p} \frac{1}{n} \log(L(X_{i}^{(n)},{\Pa}_{G}(X_{i})^{(n)})) + \lambda\vert E_G \vert
\end{align*}
where $L$ is the likelihood function \citep[cf. Definition~5.1 in][]{Nandy}. As oracle version of this scoring criterion, we use the true covariance matrix to compute the expected log-likelihood \citep[see][]{Nandy}. We denote the output of the oracle version of GES by $\GES_{\lambda}(f)$ and the output of the sample version of GES by $\GES_{\lambda}(X^{(n)}).$

\cite{ChickeringGES} showed consistency for GES for a class of scoring criteria including the Bayesian Information Criterion (BIC), which corresponds to $\lambda = \log(n)/(2n).$ The oracle version of GES is sound for $\lambda=0,$ i.e., $\GES_0(f)=\CPDAG(G_0).$

Given the density $f$ of $X$, the \emph{solution path} of the oracle version of GES is defined as the ordered set of CPDAGs $\GES_{\lambda}(f)$ for increasing values of the penalty parameter $\lambda$, $\lambda \geqslant 0$. Given $n$ i.i.d.\ samples $X^{(n)}$, the solution path of the sample version of GES is defined as the ordered set of estimated CPDAGs $\GES_{\lambda}(X^{(n)})$ for increasing values of the penalty parameter $\lambda$, for $\lambda \geqslant \log(n)/(2n)$.

\cite{Nandy} showed that the difference in score between two DAGs $G=(X,E)$ and $G^{\prime}=(X,E^{\prime})$ that differ by a single edge, i.e., $E^{\prime} = E \cup \{X_{i} \rightarrow X_{j}\},$ is given by
\begin{align} \label{Equation: Preetam}
       & S_{\lambda}(G^{\prime},X^{(n)}) - S_{\lambda}(G,X^{(n)}) \notag \\
       & \qquad \qquad = \frac{1}{2}\log(1- \hat{\rho}^{2}_{X_{i},X_{j} \vert \Pa_{G}(X_{j})}) + \lambda
\end{align}
(see Lemma~\ref{Lemma: 5.1 Preetam} of the supplementary material).
An edge is added (or deleted) in the forward (or backward) phase of GES only if this quantity is negative. To obtain the oracle version of Equation~\eqref{Equation: Preetam} we use the true covariance matrix to compute the partial correlation.


\section{AGES}\label{Section: Algorithms and theory}

The main idea behind our new algorithm, Algorithm~\ref{Algorithm: Main}, is to consider a sequence of sub-DAGs of the underlying DAG $G_{0}$, to compute their CPDAGs, and finally to aggregate these CPDAGs. Considering only sub-DAGs of $G_{0}$ ensures that if an edge is oriented in one of these CPDAGs it has the same orientation as in $G_{0}$. This property makes the aggregation intuitive since all CPDAGs will have compatible edge orientations. To learn these CPDAGs we need to assume a special type of $\delta$-strong faithfulness with respect to the sub-DAGs (see Theorem~\ref{Theorem: Equality of the CPDAGs}). The CPDAGs mentioned above can be computed efficiently using GES (see Section~\ref{Section: Computation}). Therefore, we base our new algorithm on GES and call it aggregated GES (AGES).

\subsection{THE APDAG $A_{0}$}\label{Section: New target}

\IncMargin{1em}
\begin{algorithm}[t]
   \SetKwInOut{Input}{input}
   \SetKwInOut{Output}{output}
   \Input{Ordered set of CPDAGs $\mathcal{C}=\{C_{0},\ldots,C_{k}\}$}
   \Output{APDAG $A$}
   \caption{AggregateCPDAGs}
   \label{Algorithm: AggregateCPDAGs}
   \BlankLine
   $A \leftarrow C_{0}$\\
   \For{$i \in \{1,\ldots, k\}$}{
   Define $P \leftarrow A$\\
   \For{All edges in $C_{i}$}{
   \If{an edge is oriented in $C_i$ but not in $P$}{Orient it in $P$ as in $C_{i}$} \label{Line in AggregateCPDAGs, define aggregation}
   }
   \If{$P$ is extendible to a DAG}{\label{Line in AggregateCPDAGs}
   $A \leftarrow P$
   }
   }
   \Return MeekOrient$(A)$ (Sec.\ \ref{Section: Preliminaries SuppMat} of the supp.\ material)
\end{algorithm}
\DecMargin{1em}

We construct our new target, the aggregated PDAG (APDAG) $A_0$, with the following four steps:
\begin{enumerate}[label=\text{S.\arabic*}]
\item \label{Step 1} Given a multivariate density $f$ of $X$, compute the solution path of the oracle version of GES for $\lambda\geqslant 0$ and keep the outputs whose skeletons are contained in the skeleton of $C_{0}=\GES_{0}(f).$  This yields a set of CPDAGs $\mathcal{C} = \{C_{0},\cdots, C_{k}\}$ with associated penalty parameters $\lambda_{0}<\ldots<\lambda_{k}$.\footnote{Throughout, we use the convention that any CPDAG computed by GES is associated with the smallest possible value of the penalty parameter $\lambda$ for which this output can be obtained.}
\item \label{Step 2} Construct the set of DAGs $\mathcal{G}=\{ G_{0}, \ldots, G_{k} \}$ consisting of $G_{0}$ restricted to the skeletons of the CPDAGs in $\mathcal{C}$.
\item \label{Step 3} Construct the CPDAGs $\tilde{\mathcal{C}}=\{\tilde{C}_{0},\ldots,\tilde{C}_{k}\}$ where $\tilde{C_{i}}=\CPDAG(G_{i})$, $0\leqslant~i~\leqslant~k$.
\item \label{Step 4} Let $A_0 = $ $\AggregateCPDAGs(\tilde{\mathcal C})$ (Algorithm \ref{Algorithm: AggregateCPDAGs}).
\end{enumerate}

We emphasize that $A_0$ is a theoretical object, since its construction involves the oracle version of GES and the orientations of the true underlying DAG $G_{0}.$
The construction ensures that $G_1,\dots,G_k$ are sub-DAGs of $G_{0}.$ Hence, any oriented edges in the corresponding CPDAGs $\tilde C_1,\dots,\tilde C_k$ also correspond to those in $G_0$. As a result, the APDAG $A_{0}$ has the same skeleton as $C_0$ and
 $$\DirectedPart(C_0) \subseteq \DirectedPart(A_{0}) \subseteq \DirectedPart(G_0).$$
This makes $A_{0}$ an interesting object to investigate.\footnote{We note that the if-clause on line~\ref{Line in AggregateCPDAGs} of Algorithm~\ref{Algorithm: AggregateCPDAGs} is not needed when applying the algorithm to $\tilde{\mathcal C}$; it is needed in the context of Algorithms \ref{Algorithm: Main} and \ref{Algorithm: Main(sample)}.}

\subsection{ORACLE VERSION OF AGES AND SOUNDNESS}\label{Subsection: AGES}

The oracle version of AGES is given in pseudocode as Algorithm~\ref{Algorithm: Main}. We use Example~\ref{Example: AGES} to illustrate it. Soundness of the algorithm is shown in Theorem~\ref{Theorem: Equality of the CPDAGs}.

\IncMargin{1em}
\begin{algorithm}[t]
   \SetKwInOut{Input}{input}
   \SetKwInOut{Output}{output}
   \Input{Distribution of $X$}
   \Output{APDAG $A$}
   \caption{AGES (oracle)}
   \label{Algorithm: Main}
   \BlankLine
   Compute the solution path of the oracle version of GES for $\lambda \geqslant 0$\\
   Discard all outputs whose skeletons are not contained in the skeleton of the output when $\lambda=0.$ Denote the remaining set of CPDAGs associated with $\lambda_{0}<\ldots<\lambda_{k}$ by $\mathcal{C}=\{C_{0}, \ldots, C_{k}\}$\label{Line: Discard Oracle}\\
   \Return $\AggregateCPDAGs(\mathcal{C})$
\end{algorithm}
\DecMargin{1em}

\begin{example}\label{Example: AGES}
    Consider the density $f$ generated by the weighted DAG in Figure~\ref{Subfigure Example AGES DAG} with $\varepsilon\sim N(0,D)$, where $D$ is a diagonal matrix with entries $(0.3,0.4,0.3,0.4)$. We compute the solution path of the oracle version of GES, shown in the six CPDAGs in Figures~\ref{Subfigure Example AGES CPDAG 0} - \ref{Subfigure Example AGES CPDAG 5}, corresponding to $\lambda_{0}< \dots < \lambda_5$.
    We discard $C_1$ and $C_2$ since their skeletons are not contained in the skeleton of $C_{0}.$ We then aggregate the remaining CPDAGs $C_0,$ $C_3,$ $C_4,$ and $C_5$, using lines 1-12 of Algorithm~\ref{Algorithm: AggregateCPDAGs}. The result shown in Figure~\ref{Subfigure Example AGES APDAG No Meek} contains additional orientations, coming from the v-structure $X_1\to X_3 \leftarrow X_2$ in $C_3$ . The final output in Figure~\ref{Subfigure Example AGES APDAG} shows two further oriented edges due to MeekOrient.
\end{example}

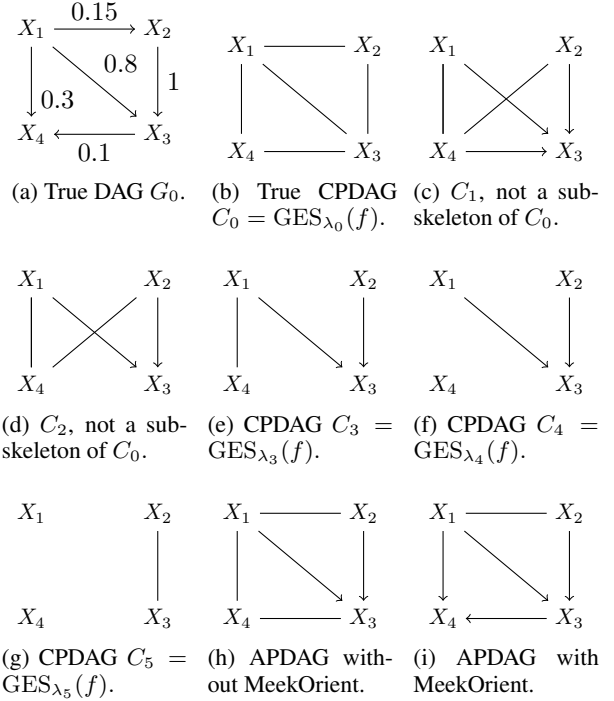
\begin{figure}[t]
   \centering
   \subfloat[True DAG $G_{0}$.]{
   \label{Subfigure Example AGES DAG}
   \begin{tikzpicture}[scale=0.7]
      \node[node] (X) at (-1.2,1) {$X_1$};
      \node[node] (Y) at (1.2,1) {$X_2$};
      \node[node] (Z) at (1.2,-1) {$X_3$};
      \node[node] (W) at (-1.2,-1) {$X_4$};
      \draw[->] (X) to [above] node{$0.15$} (Y);
      \draw[->] (X) to [above right]  node{$0.8$} (Z);
      \draw[->] (Y) to [right] node{$1$} (Z);
      \draw[->] (Z) to [below] node{$0.1$} (W);
      \draw[->] (X) to [below right] node{$0.3$} (W);
   \end{tikzpicture}
   }
   \hfill
   \subfloat[True CPDAG $C_{0}=\GES_{\lambda_{0}}(f).$]{
   \label{Subfigure Example AGES CPDAG 0}
   \begin{tikzpicture}[scale=0.7]
      \node[node] (X) at (-1.2,1) {$X_1$};
      \node[node] (Y) at (1.2,1) {$X_2$};
      \node[node] (Z) at (1.2,-1) {$X_3$};
      \node[node] (W) at (-1.2,-1) {$X_4$};
      \draw[-] (X) to [above] (Y);
      \draw[-] (X) to [left]  (Z);
      \draw[-] (Y) to [right](Z);
      \draw[-] (Z) to [below] (W);
      \draw[-] (X) to [left] (W);
   \end{tikzpicture}
   }
   \hfill
   \subfloat[$C_1$, not a sub-skeleton of $C_{0}.$]{
   \label{Subfigure Example AGES not sub-skeleton 1}
   \begin{tikzpicture}[scale=0.7]
      \node[node] (X) at (-1.2,1) {$X_1$};
      \node[node] (Y) at (1.2,1) {$X_2$};
      \node[node] (Z) at (1.2,-1) {$X_3$};
      \node[node] (W) at (-1.2,-1) {$X_4$};
      \draw[-] (Y) to [above] (W);
      \draw[->] (X) to [left]  (Z);
      \draw[->] (Y) to [right](Z);
      \draw[<-] (Z) to [below] (W);
      \draw[-] (X) to [left] (W);
      \draw[-] (X) to [left] (W);
   \end{tikzpicture}
   }
   \hfill
   \subfloat[$C_2$, not a sub-skeleton of $C_{0}.$]{
   \label{Subfigure Example AGES not sub-skeleton 2}
   \begin{tikzpicture}[scale=0.7]
      \node[node] (X) at (-1.2,1) {$X_1$};
      \node[node] (Y) at (1.2,1) {$X_2$};
      \node[node] (Z) at (1.2,-1) {$X_3$};
      \node[node] (W) at (-1.2,-1) {$X_4$};
      \draw[-] (Y) to [above] (W);
      \draw[->] (X) to [left]  (Z);
      \draw[->] (Y) to [right](Z);
      \draw[-] (X) to [left] (W);
      \draw[-] (X) to [left] (W);
   \end{tikzpicture}
   }
   \hfill
   \subfloat[CPDAG $C_{3}=\GES_{\lambda_{3}}(f).$]{
   \label{Subfigure Example AGES CPDAG 3}
   \begin{tikzpicture}[scale=0.7]
      \node[node] (X) at (-1.2,1) {$X_1$};
      \node[node] (Y) at (1.2,1) {$X_2$};
      \node[node] (Z) at (1.2,-1) {$X_3$};
      \node[node] (W) at (-1.2,-1) {$X_4$};
      \draw[->] (X) to [left](Z);
      \draw[-] (X) to [left](W);
      \draw[->] (Y) to [left](Z);
   \end{tikzpicture}
   }
   \hfill
   \subfloat[CPDAG $C_{4}=\GES_{\lambda_{4}}(f).$]{
   \label{Subfigure Example AGES APDAG CPDAG 4}
   \begin{tikzpicture}[scale=0.7]
      \node[node] (X) at (-1.2,1) {$X_1$};
      \node[node] (Y) at (1.2,1) {$X_2$};
      \node[node] (Z) at (1.2,-1) {$X_3$};
      \node[node] (W) at (-1.2,-1) {$X_4$};
      \draw[->] (X) to [left](Z);
      \draw[->] (Y) to [left](Z);
   \end{tikzpicture}
   }
   \hfill
   \subfloat[CPDAG $C_{5}=\GES_{\lambda_{5}}(f).$]{
   \label{Subfigure Example AGES CPDAG 5}
   \begin{tikzpicture}[scale=0.7]
      \node[node] (X) at (-1.2,1) {$X_1$};
      \node[node] (Y) at (1.2,1) {$X_2$};
      \node[node] (Z) at (1.2,-1) {$X_3$};
      \node[node] (W) at (-1.2,-1) {$X_4$};
      \draw[-] (Y) to [left](Z);
   \end{tikzpicture}
   }
   \hfill
   \subfloat[APDAG without MeekOrient.]{
   \label{Subfigure Example AGES APDAG No Meek}
   \begin{tikzpicture}[scale=0.7]
      \node[node] (X) at (-1.2,1) {$X_1$};
      \node[node] (Y) at (1.2,1) {$X_2$};
      \node[node] (Z) at (1.2,-1) {$X_3$};
      \node[node] (W) at (-1.2,-1) {$X_4$};
      \draw[-] (X) to [above] (Y);
      \draw[->] (X) to [above right] (Z);
      \draw[->] (Y) to [right] (Z);
      \draw[-] (Z) to [below] (W);
      \draw[-] (X) to [left](W);
   \end{tikzpicture}
   }
   \hfill
   \subfloat[APDAG with MeekOrient.]{
   \label{Subfigure Example AGES APDAG}
   \begin{tikzpicture}[scale=0.7]
      \node[node] (X) at (-1.2,1) {$X_1$};
      \node[node] (Y) at (1.2,1) {$X_2$};
      \node[node] (Z) at (1.2,-1) {$X_3$};
      \node[node] (W) at (-1.2,-1) {$X_4$};
      \draw[-] (X) to [above] (Y);
      \draw[->] (X) to [above right] (Z);
      \draw[->] (Y) to [right] (Z);
      \draw[->] (Z) to [below] (W);
      \draw[->] (X) to [left](W);
   \end{tikzpicture}
   }
   \caption{Illustration of the oracle AGES algorithm (see Example~\ref{Example: AGES}).}
   \label{Figure: Example AGES}
\end{figure}

\begin{theorem}\label{Theorem: Equality of the CPDAGs}
   Given a multivariate Gaussian distribution of $X$ with a perfect map $G_{0}=(X,E)$, let $\mathcal{G}$ be the set of DAGs constructed in Step~\ref{Step 2}, and let $\tilde{\mathcal{C}}$ be the corresponding set of CPDAGs of Step~\ref{Step 3}. Assume that for all $1\leqslant i \leqslant k$ the distribution of $X$ is $\delta_{i}$-strong faithful with respect to $G_{i} \in \mathcal{G}$, where $\delta_{i}$ is such that $\lambda_{i} = -1/2\log(1-\delta_{i}^2).$
   Then $\GES_{\lambda_i}(f) = C_{i}= \tilde{C}_i$ for all $1\leqslant i \leqslant k,$ and the oracle version of AGES returns the APDAG $A_{0}.$
\end{theorem}

Since the above $\delta_i$-strong faithfulness assumption with respect to $G_i$ for $1\leqslant i\leqslant k$ is related to the solution path of GES, we refer to it as \emph{path strong faithfulness}.

\subsection{SAMPLE VERSION OF AGES AND CONSISTENCY}\label{Subsection: Sample Version}

The sample version of AGES is given in Algorithm~\ref{Algorithm: Main(sample)}. We see that the algorithm considers the output of the sample version of GES for all $\lambda \geqslant \log(n)/(2n)$, i.e., by penalizing equally strong or stronger than BIC for model complexity.

\IncMargin{1em}
\begin{algorithm}[t]
   \SetKwInOut{Input}{input}
   \SetKwInOut{Output}{output}
   \Input{$X^{(n)}$, containing $n$ i.i.d.\ observations of $X$}
   \Output{Estimated APDAG $\hat A$}
   \caption{AGES (sample)}
   \label{Algorithm: Main(sample)}
   \BlankLine
   Compute the solution path of the sample version of GES for $\lambda \geqslant \log(n)/(2n)$\\
   Discard all outputs whose skeletons are not contained in the skeleton of the output when $\lambda=\log(n)/(2n)$. Denote the remaining CPDAGs, ordered according to increasing penalty parameter $\lambda$, by $\hat{\mathcal{C}}=\{\hat{C}_{0}, \ldots, \hat{C}_{k}\}$\label{Line: Discard Sample}\\
   \Return AggregateCPDAGs$(\hat{\mathcal{C}})$
\end{algorithm}
\DecMargin{1em}

In line~\ref{Line: Discard Sample} of Algorithm~\ref{Algorithm: Main(sample)}, we may obtain CPDAGs with conflicting orientations. Because of such possible conflicts, we need the if-clause on line~\ref{Line in AggregateCPDAGs} of Algorithm~\ref{Algorithm: AggregateCPDAGs}. The aggregation algorithm is constructed so that orientations in $\hat C_\ell$ are only taken into account if they are compatible with the aggregated graph based on $\hat C_0,\dots,\hat C_{\ell-1}$. In particular, the algorithm ensures that we stay within the Markov equivalence class defined by $\hat C_0=\GES_{\log(n)/(2n)}(X^{(n)})$, i.e., the output of GES.

Let $\AGES(X^{(n)})$ denote the output of AGES based on a sample $X^{(n)}$. Theorem~\ref{Theorem: Consistency} shows consistency of AGES.

\begin{theorem}\label{Theorem: Consistency}
  Under the conditions of Theorem~\ref{Theorem: Equality of the CPDAGs}, we have $$\lim_{n \rightarrow \infty}\mathbb{P}\left(\AGES(X^{(n)})=A_{0} \right) \rightarrow 1.$$
\end{theorem}

\subsection{THE PATH STRONG FAITHFULNESS ASSUMPTION}
\label{Section: The strong-faithfulness assumption}

The $\delta$-strong faithfulness assumption has been used before, for example to
prove uniform consistency and high-dimensional consistency
of structure learning methods \citep{Kalisch, Zhang}. On the other hand,
it has been criticised for being too strong \citep{StrongF}.

We do not assume the classical $\delta$-strong faithfulness for the underlying distribution with respect to $G_0$. Instead, we assume $\delta_i$-strong faithfulness of the distribution of $X$ with respect to the sequence of sub-DAGs $G_1,\dots,G_k$ as defined in Step~\ref{Step 2}, with corresponding $\lambda_1< \dots < \lambda_k$. Hence, the corresponding $\delta_i$s satisfy $\delta_1 < \dots < \delta_k$. Since smaller values of $\lambda$ typically yield denser graphs, it follows that for smaller values of $\delta_i$, the assumption has to hold with respect to a denser graph, while for larger values of $\delta_i$, the assumption has to hold with respect to a sparser graph.

\begin{example}\label{example: strong-faithfulness}
   We first analyse the path strong faithfulness assumption by considering the SEM given in Example~\ref{Example: Motivation}, but with unspecified edge weights $B_{13}$ and $B_{23}$:
   \begin{align*}
      X_1 &= \varepsilon_{1}\\
      X_2 &=  0.1 \cdot X_1 + \varepsilon_{2}\\
      X_3 &= B_{13}X_1 + B_{23}X_2 + \varepsilon_{3},
   \end{align*}
   and $\varepsilon \sim N(0,I)$.

   Depending on the edge weights, $A_0$ can be either the APDAG in Figure~\ref{Subfigure Informative APDAG} or in Figure~\ref{Subfigure Uninformative APDAG}. Figure~\ref{Figure dependence on edge weights} illustrates how $A_0$ and the path strong faithfulness assumption are related to the edge weights $B_{13} \in [-2,2]$ and $B_{23} \in [-2,2]$. We split the $[-2,2]\times[-2,2]$ rectangle into the following three regions:
   \begin{description}
      \item[White region:] $A_0$ equals the APDAG in Figure~\ref{Subfigure Informative APDAG} and the path strong faithfulness assumption is satisfied.
      \item[Grey region:] $A_0$ equals the APDAG in Figure~\ref{Subfigure Uninformative APDAG} and the path strong faithfulness assumption is satisfied.
      \item[Black region:] $A_0$ equals the APDAG in Figure~\ref{Subfigure Uninformative APDAG} and the path strong faithfulness assumption is violated.
   \end{description}

   The output $A$ of the oracle version of AGES can be one of the four APDAGs in Figure \ref{Figure all APDAGs}. Theorem \ref{Theorem: Equality of the CPDAGs} guarantees that $A = A_0$ when path strong faithfulness is satisfied, i.e., outside of the black region. Further, for this example, $A$ equals one of the APDAGs in Figures \ref{Subfigure APDAG with one wrong orientation} and \ref{Subfigure APDAG with two wrong orientations} on the black region. This demonstrates that, in this simple example with $p=3$, our strong faithfulness assumption is, in fact, a necessary and sufficient condition for having $A= A_0$. We emphasize that for $p>3$, we may have $A= A_0$ even when the strong faithfulness assumption is violated.

   Figure~\ref{Figure dependence on edge weights} shows that in a large fraction of the plane we gain structural information (white region), on a smaller part we perform as GES (grey region), and on another smaller part we make some errors when orienting edges (black region). Details about the construction of Figure~\ref{Figure dependence on edge weights} are given in Section~\ref{Section: path strong faithful} of the supplementary material.
\end{example}

\begin{figure}[t]
   \centering
   \hspace{1cm}
   \subfloat[Informative APDAG.]{
   \label{Subfigure Informative APDAG}
   \begin{tikzpicture}[scale=0.7]
      \node[node] (X) at (-1.2,1) {$X_1$};
      \node[node] (Y) at (1.2,1) {$X_2$};
      \node[node] (Z) at (0,-1) {$X_3$};
      \draw[-] (X) to [above]  (Y);
      \draw[thick,->] (X) to [left]  (Z);
      \draw[thick,->] (Y) to [right] (Z);
   \end{tikzpicture}
   }
   \hfill
   \subfloat[Uninformative APDAG.]{
   \label{Subfigure Uninformative APDAG}
   \begin{tikzpicture}[scale=0.7]
      \node[node] (X) at (-1.2,1) {$X_1$};
      \node[node] (Y) at (1.2,1) {$X_2$};
      \node[node] (Z) at (0,-1) {$X_3$};
      \draw[-] (X) to [above]  (Y);
      \draw[-] (X) to [left]  (Z);
      \draw[-] (Y) to [right]  (Z);
   \end{tikzpicture}
   }
   \hspace{1cm}
   \vfill
   \hspace{1cm}
   \subfloat[APDAG with one wrong orientation.]{
   \label{Subfigure APDAG with one wrong orientation}
   \begin{tikzpicture}[scale=0.7]
      \node[node] (X) at (-1.2,1) {$X_1$};
      \node[node] (Y) at (1.2,1) {$X_2$};
      \node[node] (Z) at (0,-1) {$X_3$};
      \draw[thick,->] (X) to [above]  (Y);
      \draw[thick,->] (Z) to [left]  (Y);
      \draw[-] (X) to [right] (Z);
   \end{tikzpicture}
   }
   \hfill
   \subfloat[APDAG with two wrong orientations.]{
   \label{Subfigure APDAG with two wrong orientations}
   \begin{tikzpicture}[scale=0.7]
      \node[node] (X) at (-1.2,1) {$X_1$};
      \node[node] (Y) at (1.2,1) {$X_2$};
      \node[node] (Z) at (0,-1) {$X_3$};
      \draw[thick,->] (Y) to [above]  (X);
      \draw[thick,->] (Z) to [left]  (X);
      \draw[-] (Y) to [right] (Z);
   \end{tikzpicture}
   }
   \hspace{1cm}
   \hfill
   \caption{The possible outputs of AGES in Example~\ref{example: strong-faithfulness}.}
   \label{Figure all APDAGs}
\end{figure}
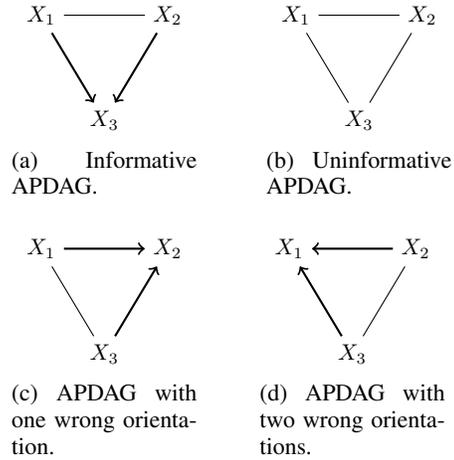

\begin{figure}[t]
   \centering
   \includegraphics[width=0.35\textwidth,height=0.35\textwidth]{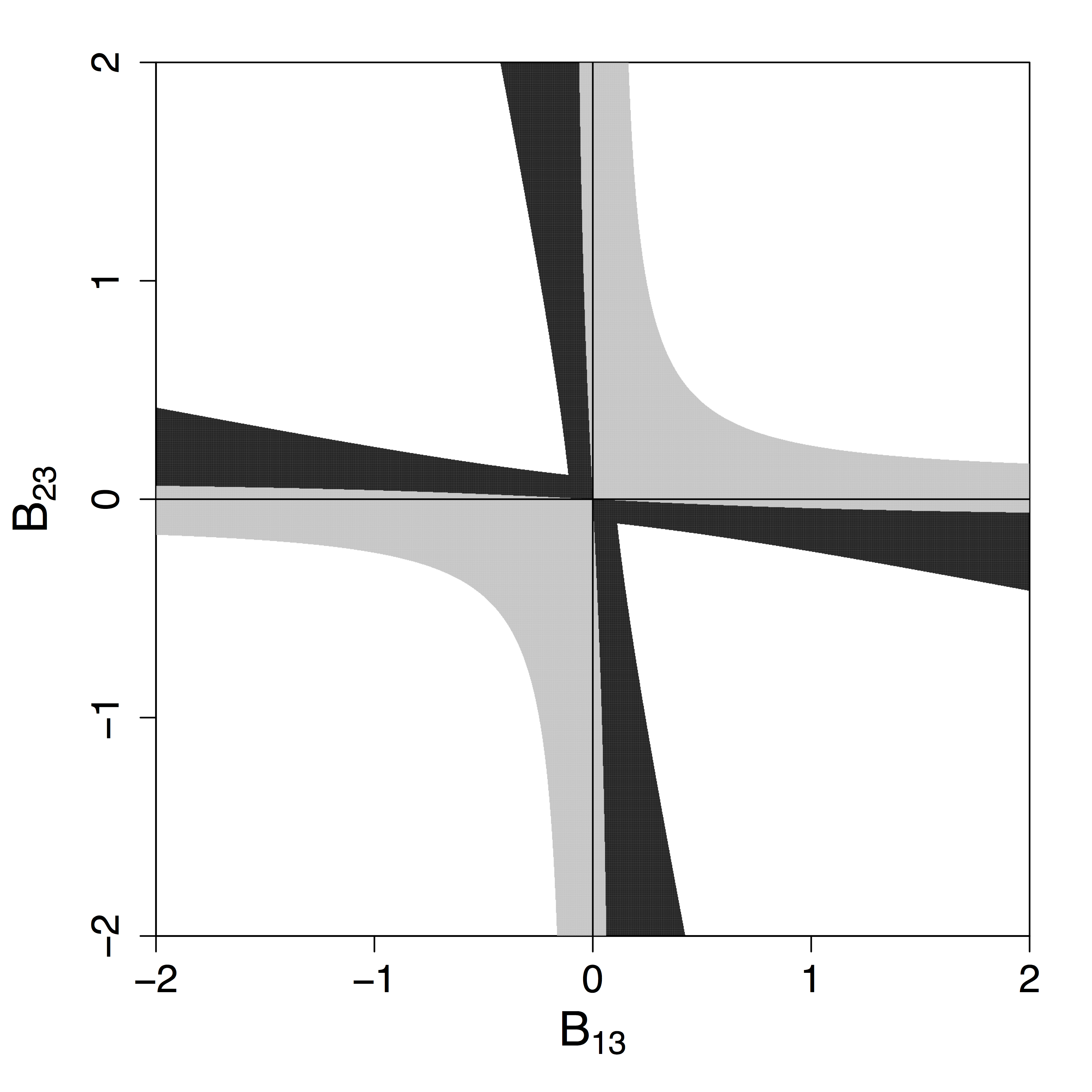}
   \caption{Visual representation of the dependence of $A_0$ and the path strong faithfulness assumption on the edge weights in Example~\ref{example: strong-faithfulness}.}
   \label{Figure dependence on edge weights}
\end{figure}

The path strong faithfulness assumption is sufficient but not necessary for Theorem~\ref{Theorem: Equality of the CPDAGs}. In Section~\ref{Subsection: Strong Faithfulness} of the supplementary material we provide a weaker version of the assumption that is necessary and sufficient for equality of $\tilde{\mathcal C}$ (as defined in Step~\ref{Step 3}) and $\mathcal C$ (as defined in line~\ref{Line: Discard Oracle} of Algorithm \ref{Algorithm: Main}). This weaker version is only sufficient for equality of the true APDAG $A_0$ and the oracle output $A$ of AGES $(\AggregateCPDAGs(\mathcal C))$, since not all orientations of the CPDAGs in $\mathcal C$ are used in the aggregation process. The supplementary material also contains empirical results where we evaluated equality of $\tilde{\mathcal C}$ and $\mathcal C$, as well $A_0$ and $A$, for the simulation setting described in Section \ref{Section: Simulations}.

\subsection{COMPUTATION}\label{Section: Computation}

The forward phase of GES (of both the oracle version and the sample version) can be computed at once for all $\lambda\geqslant 0$. This follows from Equation~\eqref{Equation: Preetam}.
At each step in the forward phase, GES conceptually searches for the in absolute value largest partial correlation $|\rho_{X_i,X_j|\Pa_G(X_j)}|$ among all DAGs $G$ in the current Markov equivalence class, and all pairs $X_i$ and $X_j$ that are not adjacent in $G$ and where $X_i$ is a non-descendant of $X_j$ in $G$. The algorithm then adds the corresponding edge $X_i\to X_j$ to $G$ if the score is improved, that is, if $1/2\log(1-\rho_{X_{i},X_{j}\vert Pa_G(X_{j})}^2) + \lambda < 0$, and then constructs the CPDAG the resulting DAG.

Thus, starting the forward phase with the empty graph and a very large $\lambda$, no edge is added. By decreasing $\lambda$ so that $\lambda < \max_{i,j} -1/2\log(1-\rho_{X_{i},X_{j}}^2)$, the first edge is added. By decreasing $\lambda$ further, one can compute the entire solution path of the forward phase in one go, analogously to the computation of the solution path of the lasso \citep{Lasso, LassoSolutionPath}.

For each distinct output of the forward phase, obtained for a given $\lambda$, one has to run the backward phase with this $\lambda$. Since the backward phase of GES usually only conducts very few steps, this does not cause a large computational burden.

The fast computation of the entire solution path of GES is one of the reasons for basing our approach on GES, rather than, for example, on the PC-algorithm for a range of different tuning parameters $\alpha$.


\section{EMPIRICAL RESULTS}\label{Section: Simulations}

\subsection{SIMULATION SETUP}\label{subsection: data}

We simulate data from SEMs of the following form:
$$X = B^{T}X + \varepsilon,$$
with $\varepsilon \sim \mathcal{N}(0,D)$, where $D$ is a $p\times p$ diagonal matrix whose diagonal entries are drawn independently from a Unif(0.5,1.5) distribution.

In order to vary the concentration of strong and weak edge weights as well as the sparsity of the models, we consider all combinations of pairs $(q_s,q_w) \in \{0.1, 0.3, 0.5, 0.7\}$ such that $q_s+q_w \leqslant 1$. Each entry of the matrix $B$ has a probability of $q_s$ of being strong, of $q_w$ of being weak, and of $(1-q_s-q_w)$ of being $0$. The nonzero edge weights in the $B$ matrix are drawn independently as follows: the absolute values of the weak and the strong edge weights are drawn from Unif(0.1,0.3) and Unif(0.8,1.2), respectively. The sign of each edge weight is chosen to be positive or negative with equal probabilities. Finally, in order to investigate whether our algorithm performs at least as good as GES when we do not encourage the presence of weak edges, we also simulate from SEMs with $q_s\in \{0.1,0.2,\ldots, 1\}$ and $q_w=0$.

We simulate from SEMs with $p=10$ variables. The sample size used in the plots in the main paper is $10000.$ The number of simulations for each settings is 500.

In Section~\ref{Section: other empirical results} of the supplementary material we show additional plots corresponding to sample sizes $100$ and $1000$. Those plots show a similar pattern as the ones in the main paper, but the ability to gain additional edge orientations diminishes for smaller $n$. Section~\ref{Section: Simulations with p=100} of the supplementary material also shows simulation results for $p=100$ and varying sample sizes.

\subsection{SIMULATION RESULTS}\label{Subsection: Performance Assessment}

Since AGES always outputs the same skeleton as GES by construction, we analyse the performance of GES and AGES by comparing their precision and recall in estimating the directed part of the true DAG. The recall is the ratio of the number of correctly oriented edges in the estimated graph and the total number of oriented edges in the true DAG. The precision is the ratio of the number of correctly oriented edges in the estimated graph and the total number of oriented edges in the estimated graph.

Figure~\ref{Figure: recall0505} summarizes the performance of GES and AGES (with $\lambda=\log(n)/(2n)$) for all combinations of $(q_s,q_w) \in \{0.1, 0.3, 0.5, 0.7\}$ such that $q_s+q_w \leqslant 1$. In each setting, AGES outperforms GES in recall, while achieving a roughly similar performance as GES in precision. This demonstrates that AGES is able to orient more edges than GES without increasing the false discovery rate.

\begin{figure}[t]
   \centering
   \includegraphics[scale=0.25]{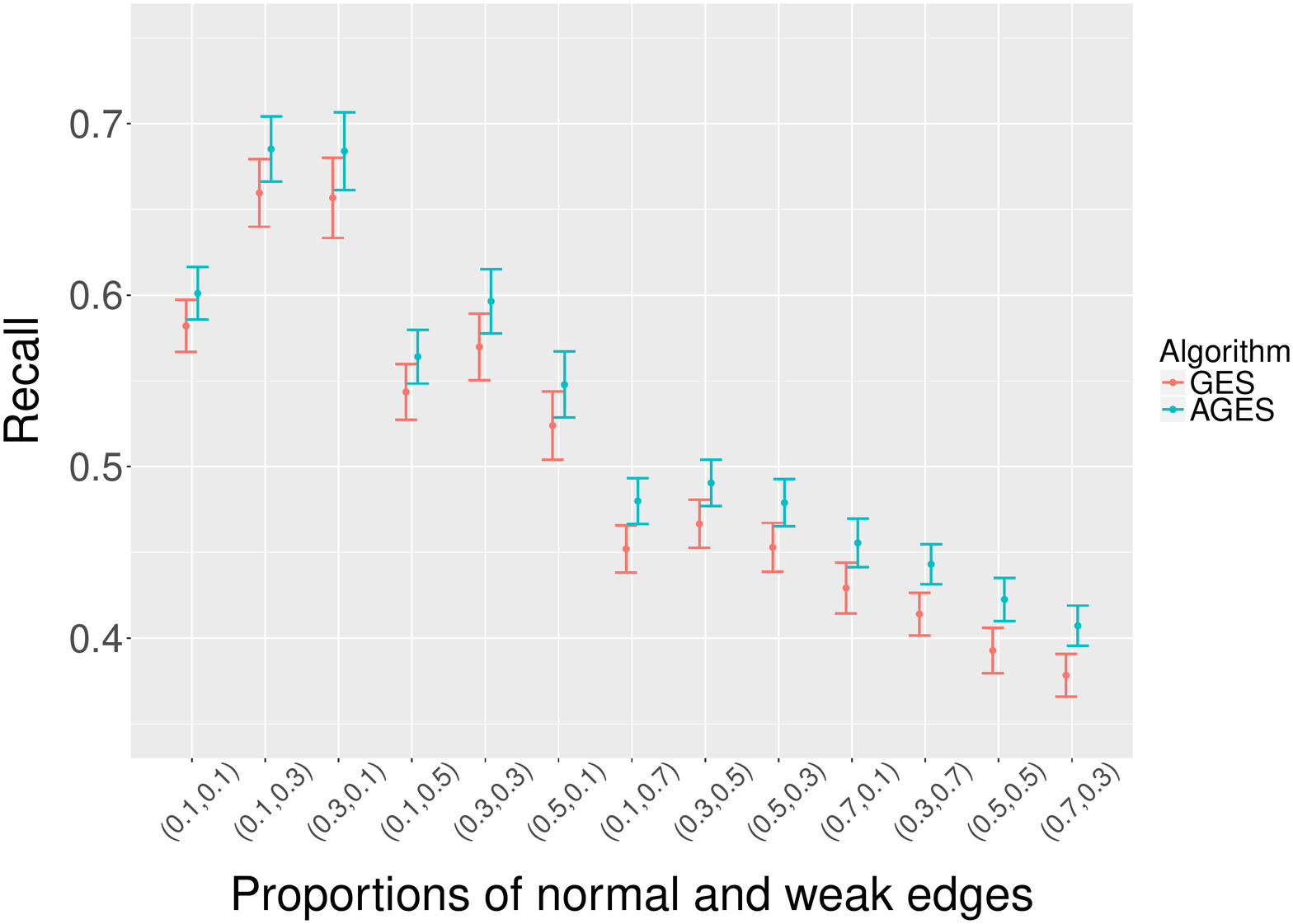}\\
   \vspace{0.5cm}
   \includegraphics[scale=0.25]{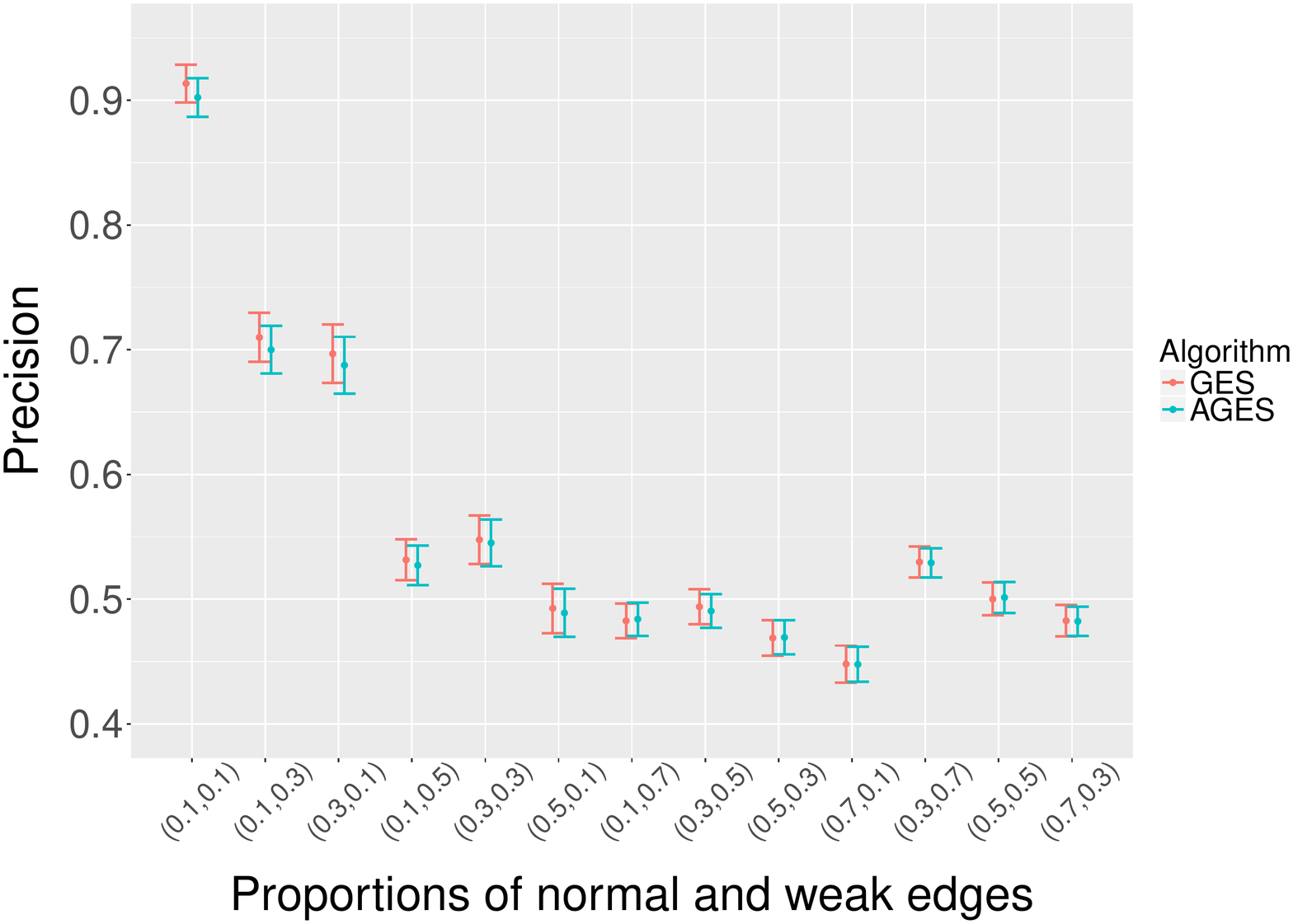}
   \caption{Mean precision and recall of GES and AGES over 500 simulations for all combinations of $(q_s,q_w) \in \{0.1, 0.3, 0.5, 0.7\}$ such that $q_s+q_w \leqslant 1$, using $\lambda=\log(n)/(2n)$ and $n=10000$ (see Section~\ref{subsection: data}). The bars in the plots correspond to $\pm$ twice the standard error of the mean.}
   \label{Figure: recall0505}
   \vspace{-.5cm}
\end{figure}

Figure~\ref{Figure: lambdas0505} compares the performance of GES and AGES for various choices of the penalty parameter $\lambda$ when $(q_s,q_w) = (0.3,0.7)$. In each case, we use the chosen penalty of GES as the minimum penalty of AGES, so that the skeletons of both outputs are identical. We see that AGES outperforms GES for all penalty parameters, and that AGES is less sensitive to the choice of the penalty parameter.

\begin{figure}[t]
   \centering
   \includegraphics[scale=0.25]{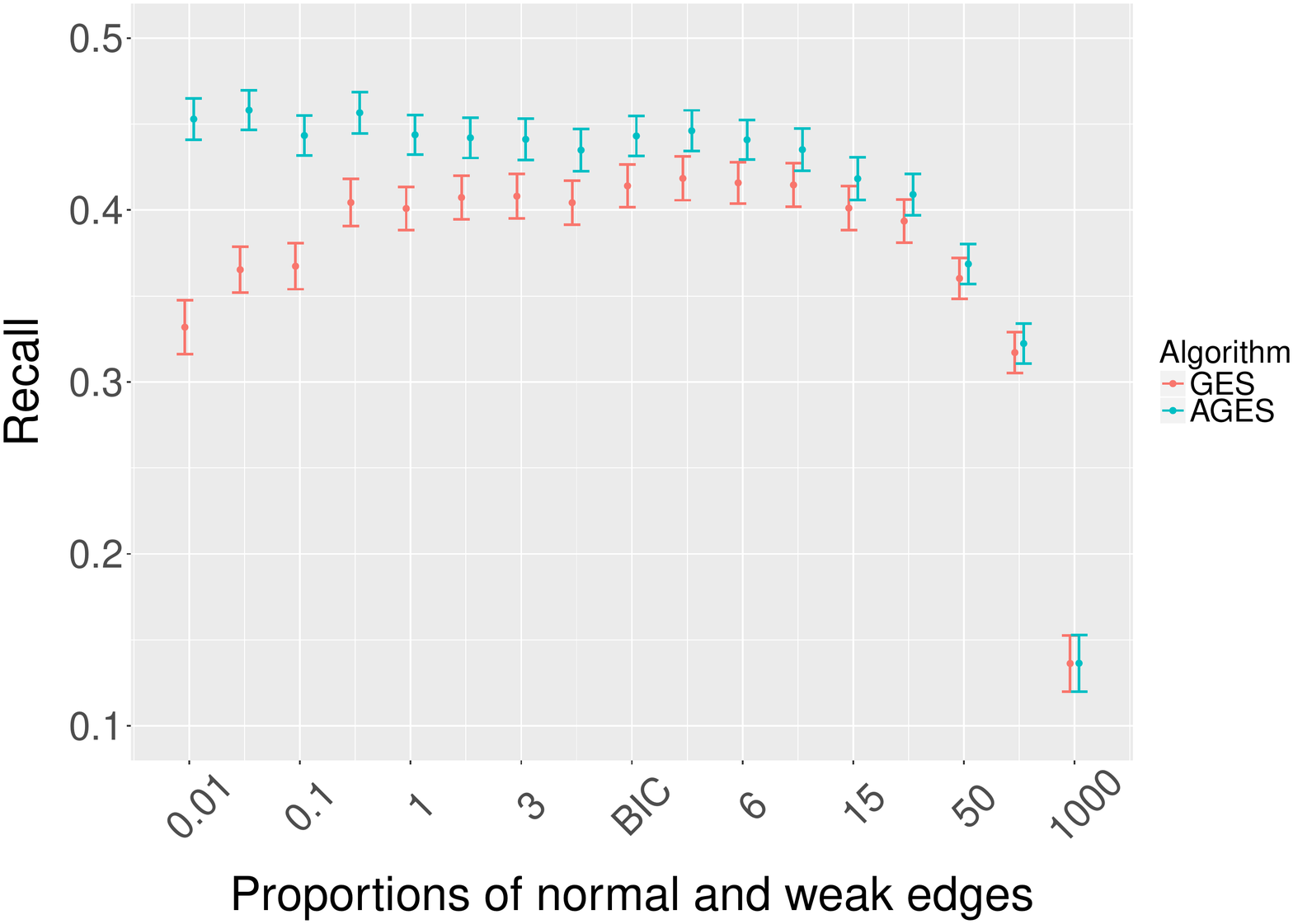}\\
   \vspace{0.5cm}
   \includegraphics[scale=0.25]{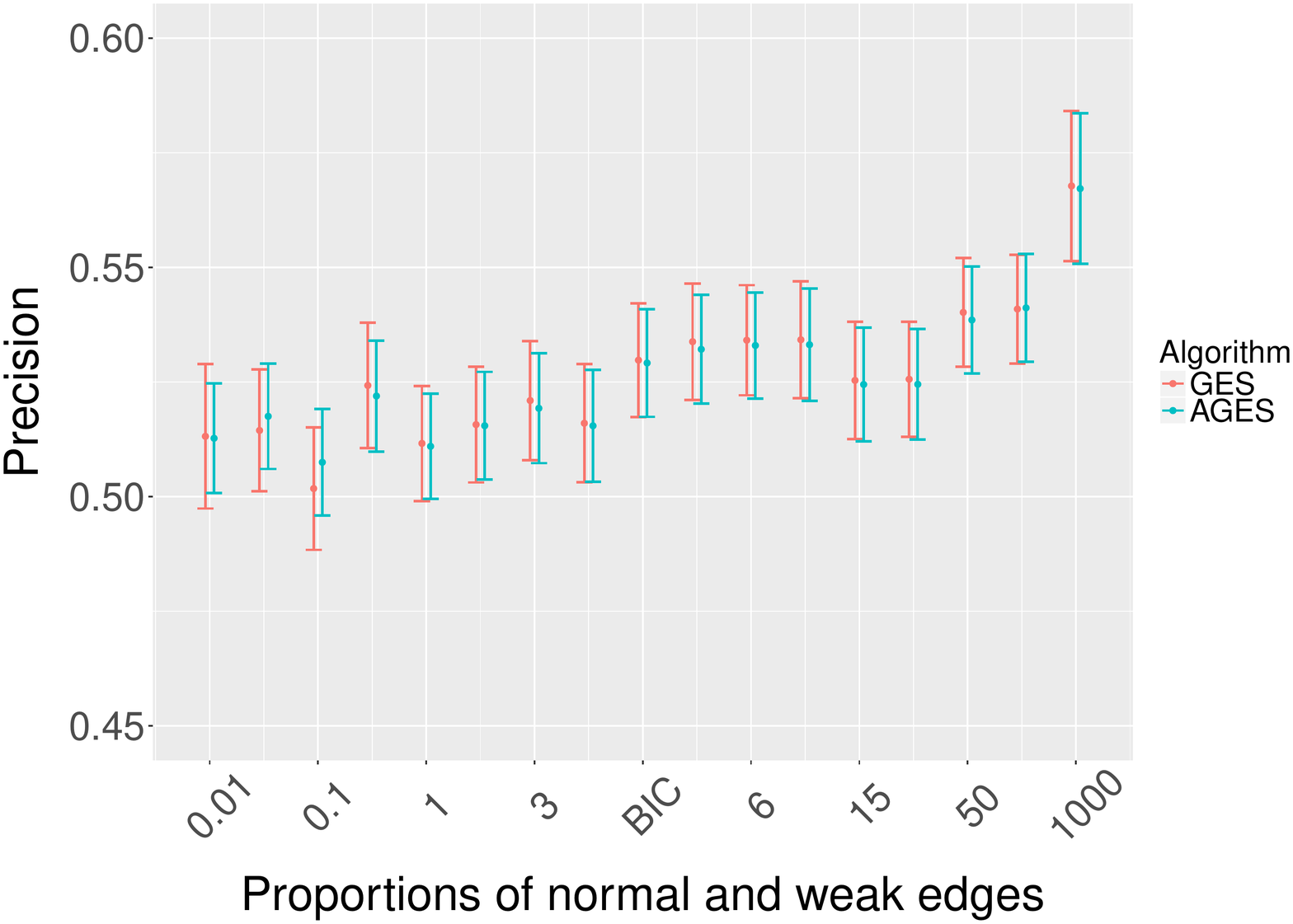}
   \caption{Mean precision and recall of GES and AGES over 500 simulations for $(q_s,q_w) = (0.3,0.7)$, using $n=10000$ and varying values of $\lambda$ (see Section~\ref{subsection: data}). The bars in the plots correspond to $\pm$ twice the standard error of the mean.}
   \label{Figure: lambdas0505}
\end{figure}

Figure~\ref{Figure: prec/recall1} compares GES and AGES for $q_s\in\{0.1,0.2,\ldots, 1\}$ and $q_w =0$, using again $\lambda=\log(n)(2n)$. We see that AGES outperforms GES in recall for all values of $q_s$. There tends to be a small loss in precision for the sparser graphs.

\begin{figure}[t]
   \centering
   \includegraphics[scale=0.25]{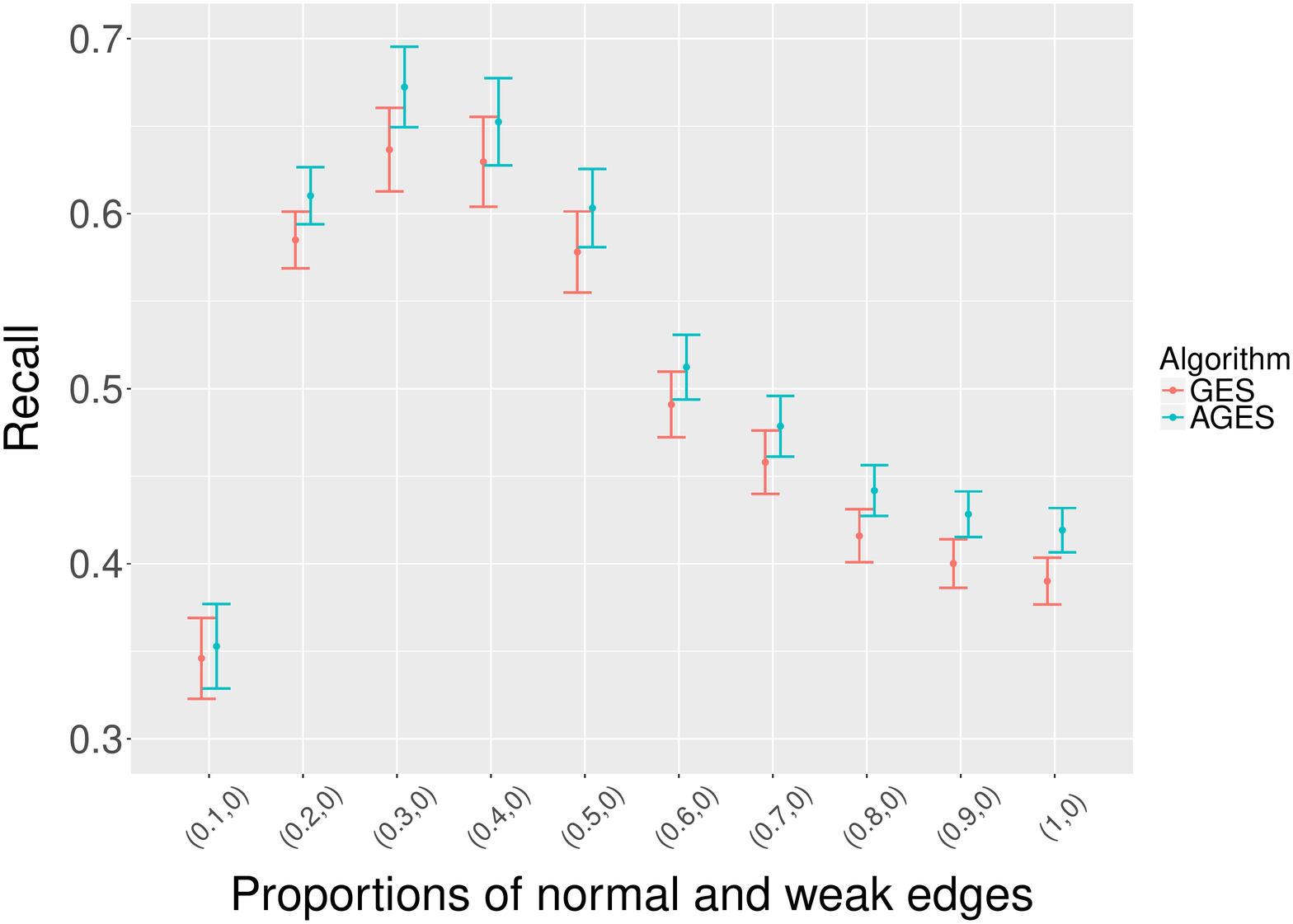}\\
   \vspace{0.5cm}
   \includegraphics[scale=0.25]{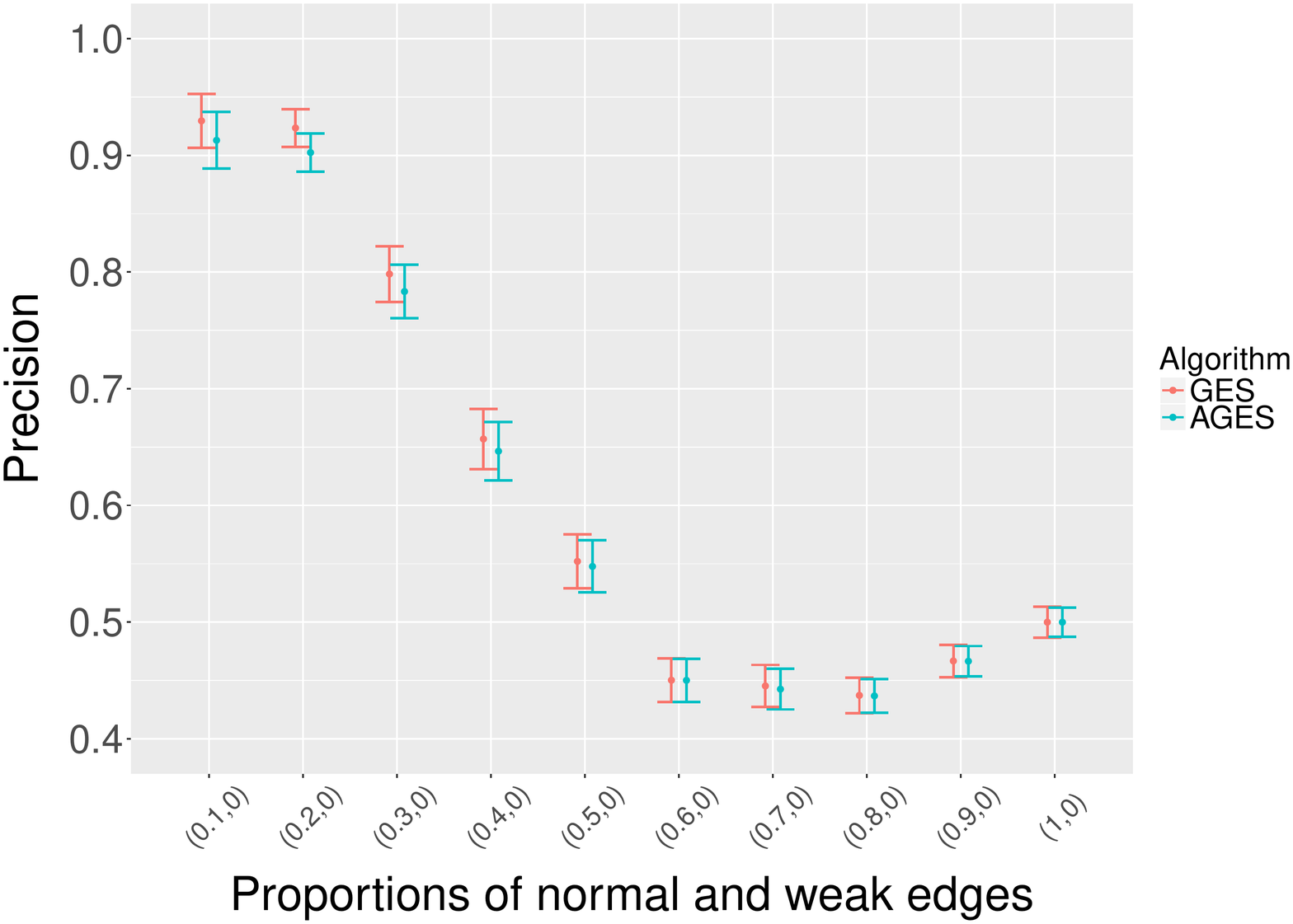}
   \caption{Mean precision and recall over 500 simulations with $q_s\in\{0.1,0.2,\ldots, 1\}$ and $q_w =0$, using $\lambda=\log(n)/(2n)$ and $n=10000$ (see Section~\ref{subsection: data}). The bars in the plots correspond to $\pm$ twice the standard error of the mean.}
   \label{Figure: prec/recall1}
\end{figure}

\subsection{APPLICATION TO SINGLE CELL DATA}\label{Section: Real data}

We apply AGES to the well-known single cell data of \cite{Sachs}, consisting of quantitative amounts of 11 proteins in human T-cells that were measured under 14 experimental conditions. In each experimental condition, different interventions were made, concerning the abundance or the activity of the molecules\footnote{An activity intervention can either activate or inhibit the molecule} \citep{Sachs,Mooij}. We analyze each experimental condition separately, yielding 14 data sets with sample sizes between 700 and 1000.

\cite{Sachs} presented a conventionally accepted signalling network for these proteins \citep[][Figures 2 and 3]{Sachs}. We use this to determine a ground truth for each experimental condition (see Section~\ref{Section: Real data suppmat} of the supplementary material), so that we can assess the performance of AGES in comparison to GES on these data.

Again, since the skeletons of the outputs of GES and AGES are identical by construction, we only evaluate the directed edges. Moreover, we limit ourselves to adjacencies that are present in the true network.
Considering these adjacencies, AGES found additional edge orientations in 6 experimental conditions.  Table~\ref{Figure: Tab} summarizes the results. In experimental conditions 8 and 9, AGES was able to substantially improve the output, while in the other 4 conditions (4, 5, 13 and 14), AGES and GES had roughly similar performances.
Thus, although these data almost certainly violate various assumptions of our methods (acyclicity, Gaussianity, path strong faithfulness, hidden confounders), we obtain encouraging results.

\begin{table}
   \begin{center}
   \begin{tabular}{ | c | c | c | c | c  | c | c | }
       \hline
       Experimental condition & 4 & 5 & 8 & 9 & 13 & 14 \\ \hline
       Correct & 0 & 1 & 8 & 5  & 1 & 1 \\ \hline
       Wrong 	& 1 & 0 & 0 & 0 & 2 & 0 \\
       \hline
     \end{tabular}
   \end{center}
   \caption{For each of the listed experimental conditions, we report the number of correct and wrong edge orientations among edge orientations that were found by AGES but not by GES. The results are limited to  adjacencies that are present in the true network (see Figure~\ref{Figure: PDAG real data} of the supplementary material), and correctness of edge orientations was evaluated with respect to this network. }
   \label{Figure: Tab}
   \vspace{-.5cm}
\end{table}


\section{DISCUSSION}\label{Section: Discussion}

We considered structure learning of linear Gaussian SEMs with weak edges. We presented a new graphical object, called APDAG, that aggregates the structural information of many CPDAGs, yielding additional orientation information. We proposed a structure learning algorithm that uses the solution path of GES to learn this new object and gave sufficient conditions for its soundness and consistency. The algorithm will be made available in the R-package \texttt{pcalg} \citep{Kalischpcalg}.

We applied AGES in a simulation study and on data from \cite{Sachs}. Despite the fact that in both cases the assumptions of Theorem~\ref{Theorem: Equality of the CPDAGs} are likely violated, we obtained promising results.

Our work can be easily extended to the so called nonparanormal distributions \citep{LiuEtAl09,Harris2013}. In this setting we assume that there is a latent linear Gaussian SEM and that each observed variable is a strictly increasing (or strictly decreasing) transformation of the corresponding latent variable. In this case, the weakness of an edge can be connected to its edge weight in the latent linear Gaussian SEM and we can use AGES with a rank correlation based scoring criterion as defined in \cite{Nandy}.

Moreover, the Gaussian error assumption can be dropped, i.e., we can consider linear SEMs with arbitrary error distributions. This is due to a one-to-one correspondence between zero partial correlations in a linear SEM with arbitrary error distributions and d-separations in its corresponding DAG \citep[e.g.,][]{hoyer08}. When all error variables are non-Gaussian, one can use the LiNGAM algorithm \citep{ShimizuEtAl06-JMLR} to recover the data generating DAG uniquely. In this case, one would therefore not run GES or AGES. If some error variables are Gaussian and others are non-Gaussian, \cite{hoyer08} proposed a combination of PC and LiNGAM. It would be an interesting direction for future work to combine (A)GES with LiNGAM for a mixture of Gaussian and non-Gaussian error variables.

\subsection{Acknowledgements}
This work was supported in part by the Swiss NSF Grant 200021\_172603.


\newpage

\setcounter{section}{0}
\setcounter{figure}{0}
\setcounter{table}{0}

\section*{SUPPLEMENT}

This is the supplement of the paper ``Structure Learning of Linear Gaussian Structural Equation Models with Weak Edges'', which  we refer to as the ``main paper''.

\section{PRELIMINARIES}\label{Section: Preliminaries SuppMat}

Two vertices $X_{i}$ and $X_{j}$ are \emph{adjacent} if there is an edge between them. A \emph{path} between $X_{i}$ and $X_{j}$ is a sequence $(X_{i}, \ldots, X_{j})$ of distinct vertices in which all pairs of successive vertices are adjacent. A \emph{directed path} is a path between $X_i$ and $X_j$ where all edges are directed towards $X_{j}$, i.e., $X_i \rightarrow \dots \rightarrow X_j$. A directed path from $X_{i}$ to $X_{j}$ together with the edge $X_{j} \rightarrow X_{i}$ forms a \emph{directed cycle.}
If $X_{i} \rightarrow X_{j} \leftarrow X_{k}$ is part of a path, then $X_{j}$ is a \emph{collider} on this path.

A vertex $X_{j}$ is a \emph{child} of the vertex $X_{i}$ if $X_{i} \rightarrow X_{j}.$  If there is a directed path from $X_{i}$ to $X_{j},$ $X_{i}$ is a \emph{descendant} of $X_{j}$, otherwise it is a non-descendant. We use the convention that $X_{i}$ is also a descendant of itself.

A DAG encodes conditional independence constraints through the concept of d-separation \citep{Pearl2009}.
For three pairwise disjoint subsets of vertices $A,$ $B,$ and $S$ of $X,$ $A$ is d-separated from $B$ by $S,$ $A\cid B \vert S,$ if every path between a vertex in $A$ and a vertex in $B$ is \emph{blocked} by $S$. A path between two vertices $X_{i}$ and $X_{j}$ is said to be \emph{blocked} by a set $S$ if a non-collider vertex on the path is present in $S$ or if there is a collider vertex on the path for which none of its descendants is in $S.$ If a path is not blocked it is \emph{open}.

The set of d-separation constraints encoded by a DAG $G$ is denoted by $\mathcal{I}(G)$. All DAGs in a Markov equivalence class encode the same set of d-separation constraints. Hence, for a CPDAG $C$, we let $\mathcal I(C) = \mathcal I(G)$, where $G$ is any DAG in
$C$. A DAG $G_1$ is an independence map (I-map) of a DAG $G_2$ if $\mathcal I(G_1) \subseteq \mathcal I(G_2)$, with an analogous definition for CPDAGs. A DAG $G_1$ is a perfect map of a DAG $G_2$ if $\mathcal I(G_1)=\mathcal I(G_2)$, again with an analogous definition for CPDAGs.

For the proof of Theorem~\ref{Theorem: Equality of the CPDAGs} of the main paper we make use of two lemmas of \cite{Nandy}.

\begin{lemma}\label{Lemma: deleting an important edge} (cf.\ Lemma 9.5 of the supplementary material of \cite{Nandy})
   Let $G = (X,E)$ be a DAG such that $X_i \to X_j \in E$. Let $G' = (X,E\setminus\{X_i \to X_j\})$. If $G$ is an I-map of a DAG $G_1$
   but $G'$ is not, then $X_i \ncid_{G_1} X_j \mid \Pa_{G'}(X_i)$.
\end{lemma}

\begin{lemma}(cf.\ Lemma 5.1 of \cite{Nandy})\label{Lemma: 5.1 Preetam}
Let $G=(X,E)$ be a DAG such that $X_{i}$ is neither a descendant nor a parent of $X_{j}.$ Let $G^{\prime}=(X,E \cup \{ X_{i} \rightarrow X_{j} \}).$ If the distribution of $X$ is multivariate Gaussian, then the $\ell_{0}-$penalized log-likelihood score difference between $G^{\prime}$ and $G$ is\begin{multline*} S_{\lambda}(G^{\prime},X^{(n)}) - S_{\lambda}(G,X^{(n)}) \\ =\frac{1}{2}\log(1- \hat{\rho}^{2}_{X_{i},X_{j} \vert \Pa_{G}(X_{j})}) + \lambda.\end{multline*}
\end{lemma}

The last step of Algorithm~\ref{Algorithm: AggregateCPDAGs} of the main paper consists of MeekOrient. This step applies iteratively and sequentially the four rules depicted in Figure~\ref{Figure: Meek}. These orientation rules can lead to some additional orientations, and the resulting output is a maximally oriented PDAG \citep{Meek}. For an example of its utility see Example~\ref{Example: AGES} and Figure~\ref{Figure: Example AGES} in the main paper.

\begin{figure}[tb]
    \centering
    \hspace{0.5cm}
    \subfloat{
    \begin{tikzpicture}[scale=0.7]
       \node[node] (X) at (-1,1.5) {};
       \node[node] (Y) at (-1,0) {};
       \node[node] (Z) at (0.5,0) {};
       \draw[thick,->] (X) to(Y);
       \draw[-,blue] (Y) to (Z);
       \node[node] (Arrow) at (1,0.95) {R1};
       \node[node] (Arrow) at (1,0.55) {$\Rightarrow$};
       \node[node] (X2) at (1.5,1.5) {};
       \node[node] (Y2) at (1.5,0) {};
       \node[node] (Z2) at (3,0) {};
       \draw[thick,->] (X2) to(Y2);
       \draw[thick,->,blue] (Y2) to (Z2);
    \end{tikzpicture}
    }
    \hfill
    \subfloat{
    \begin{tikzpicture}[scale=0.7]
       \node[node] (X) at (-1,1.5) {};
       \node[node] (Y) at (-1,0) {};
       \node[node] (Z) at (0.5,0) {};
       \draw[thick,->] (Y) to(X);
       \draw[thick,->] (X) to(Z);
       \draw[-,blue] (Y) to (Z);
       \node[node] (Arrow) at (1,0.95) {R2};
       \node[node] (Arrow) at (1,0.55) {$\Rightarrow$};
       \node[node] (X2) at (1.5,1.5) {};
       \node[node] (Y2) at (1.5,0) {};
       \node[node] (Z2) at (3,0) {};
       \draw[thick,->] (Y2) to(X2);
       \draw[thick,->] (X2) to(Z2);
       \draw[thick,->,blue] (Y2) to (Z2);
    \end{tikzpicture}
    }
    \hspace{0.5cm}
    \vfill
    \hspace{0.5cm}
    \subfloat{
    \begin{tikzpicture}[scale=0.7]
       \node[node] (X) at (-1,1.5) {};
       \node[node] (Y) at (-1,0) {};
       \node[node] (Z) at (0.5,0) {};
       \node[node] (W) at (0.5,1.5) {};
       \draw[-] (X) to (Y);
       \draw[thick,->] (Y) to (Z);
       \draw[-] (X) to (W);
       \draw[thick,->] (W) to (Z);
       \draw[-,blue] (X) to (Z);
       \node[node] (Arrow) at (1,0.95) {R3};
       \node[node] (Arrow) at (1,0.55) {$\Rightarrow$};
       \node[node] (X2) at (1.5,1.5) {};
       \node[node] (Y2) at (1.5,0) {};
       \node[node] (Z2) at (3,0) {};
       \node[node] (W2) at (3,1.5) {};
       \draw[-] (X2) to (Y2);
       \draw[thick,->] (Y2) to (Z2);
       \draw[-] (X2) to (W2);
       \draw[thick,->] (W2) to (Z2);
       \draw[thick,->,blue] (X2) to (Z2);
    \end{tikzpicture}
    }
    \hfill
    \subfloat{
    \begin{tikzpicture}[scale=0.7]
       \node[node] (X) at (-1,1.5) {};
       \node[node] (Y) at (-1,0) {};
       \node[node] (Z) at (0.5,0) {};
       \node[node] (W) at (0.5,1.5) {};
       \draw[-,blue] (X) to (Y);
       \draw[thick,<-] (Y) to (Z);
       \draw[-] (X) to (W);
       \draw[thick,->] (W) to (Z);
       \draw[-] (X) to (Z);
       \node[node] (Arrow) at (1,0.95) {R4};
       \node[node] (Arrow) at (1,0.55) {$\Rightarrow$};
       \node[node] (X2) at (1.5,1.5) {};
       \node[node] (Y2) at (1.5,0) {};
       \node[node] (Z2) at (3,0) {};
       \node[node] (W2) at (3,1.5) {};
       \draw[thick,->,blue] (X2) to (Y2);
       \draw[thick,<-] (Y2) to (Z2);
       \draw[-] (X2) to (W2);
       \draw[thick,->] (W2) to (Z2);
       \draw[-] (X2) to (Z2);
    \end{tikzpicture}
    }
    \hspace{0.5cm}
    \caption{The four orientation rules from \cite{Meek}. If a PDAG contains one of the graphs on the left-hand-side of the four rules, then orient the blue edge as shown on the right-hand-side.}
    \label{Figure: Meek}
\end{figure}
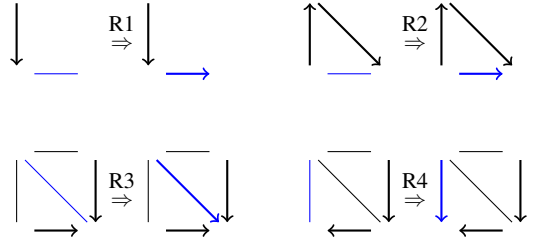

\section{PROOFS}\label{Section: Proofs}
\subsection{PROOF OF THEOREM~\ref{Theorem: Equality of the CPDAGs} OF THE MAIN PAPER}


We first establish the following Lemma.

\begin{lemma}\label{Lemma: Skeleton condition}
   Consider two CPDAGs $C_{1}$ and $C_{2}$ where $C_{1}$ is an I-map of $C_{2}.$ If $C_{1}$ and $C_{2}$ have the same skeleton, then $C_{1}$ is a perfect map of $C_{2}.$
\end{lemma}
\begin{proof}
   Let $G_1$ and $G_2$ be arbitrary DAGs in the Markov equivalence classes described by $C_1$ and $C_2$, respectively. Then $C_1$ is a perfect map of $C_2$ if and only if $G_1$ and $G_2$ have the same skeleton and the same v-structures \citep{VermaPearl}. Since $G_1$ and $G_2$ have the same skeleton by assumption, we only need to show that they have identical v-structures.

   Suppose first that there is a v-structure $X_i \to X_j \leftarrow X_k$ in $G_1$ that is not present in $G_2$. Since $G_1$ and $G_2$ have the same skeleton, this implies that $X_j$ is a non-collider on the path $(X_i,X_j,X_k)$ in $G_2$.

   We assume without loss of generality that $X_i$ is a non-descendant of $X_k$ in $G_1$.
   Then, $X_{i} \cid_{G_{1}} X_{k} \vert \Pa_{G_{1}}(X_{k})$, where $X_j \notin \Pa_{G_{1}}(X_{k})$. On the other hand, we have $X_{i} \ncid_{G_{2}} X_{k} \vert \Pa_{G_{1}}(X_{k})$, since the path $(X_{i},X_{j},X_{k})$ is open in $G_2$, since $X_{j} \notin \Pa_{G_{1}}(X_{k})$. This contradicts that $C_1$ is an I-map of $C_2$.

   Next, suppose that there is a v-structure $X_i \to X_j \leftarrow X_k$ in $G_2$ that is not present in $G_1$. Since $G_1$ and $G_2$ have the same skeleton, this implies that $X_j$ is a non-collider on the path $(X_i,X_j,X_k)$ in $G_1$.

   We again assume without loss of generality that $X_i$ is a non-descendant of $X_k$ in $G_1$. Then the path $(X_i,X_j,X_k)$ has one of the following forms: $X_i \to X_j \to X_k$ or $X_i \leftarrow X_j \to X_k$. In either case, $X_j \in Pa_{G_1}(X_k)$. Hence, $X_i \cid_{G_1} X_k | \Pa_{G_1}(X_k)$, where $X_j\in \Pa_{G_1}(X_k)$. But $X_i$ and $X_k$ are d-connected in $G_2$ by any set containing $X_j$. This again contradicts that $C_1$ is an I-map of $C_2$.
\end{proof}

\begin{proof}[Proof of Theorem~\ref{Theorem: Equality of the CPDAGs} of the main paper]

We need to prove that the CPDAGs in Step~\ref{Step 1} of the main paper and the CPDAGs in Step~\ref{Step 3} of the main paper coincide, i.e., $\mathcal{C}=\tilde{\mathcal{C}}.$
We prove this result for one of the CPDAGs. Take for instance $\tilde{C}_{\ell}$, $1\leqslant \ell \leqslant k$, the CPDAG of $G_\ell=(V,E_\ell).$
Note that $G_\ell$ is not a perfect map of the distribution of $X,$ and therefore we cannot directly use the proof of \cite{ChickeringGES2}. We can still use the main idea though, in combination with Lemma~\ref{Lemma: 5.1 Preetam}.

Consider running GES with penalty parameter $\lambda=-1/2 \log(1-\delta_{\ell}^2)$ and denote by $C^{f}$ and $C^{b}$ the output of the forward and backward phase, respectively.

Claim 1: $C^{f}$ is an I-map of $\tilde{C}_{\ell}$ i.e., all d-separation constraints true in $C^{f}$ are also true in $\tilde{C}_{\ell}.$

Proof of Claim 1:\\
Assume this is not the case, then there are two vertices $X_{i},X_{j} \in X$ and a DAG $G^{f} \in C^{f}$ such that $X_{i} \cid_{G^{f}} X_{j} \vert \{\Pa_{G^{f}}(X_{j})\setminus X_{i}\}$ but $X_{i} \ncid_{C_{\ell}} X_{j} \vert \{\Pa_{G^{f}}(X_{j})\setminus X_{i}\}.$ Because of the $\delta_{\ell}$-strong faithful condition, $\vert \rho_{X_{i},X_{j} \vert \Pa_{G^{f}}(X_{j})} \vert > \delta_{\ell}.$ Thus, adding this edge would improve the score. This is a contradiction to the GES algorithm stopping here.

Claim 2: $C^{b}$ is an I-map of $\tilde{C}_{\ell}$ i.e., all d-separation constraints true in $C^{b}$ are also true in $\tilde{C}_{\ell}.$

Proof of Claim 2:
By Claim 1 the backward phase starts with an I-map of $\tilde{C}_{\ell}.$ Suppose it ends with a CPDAG that is not an I-map of $\tilde{C}_{\ell}.$ Then, at some point there is an edge deletion which turns a DAG $G$ that is an I-map of $G_{\ell}$ into a DAG $G^{\prime}$ that is no longer an I-map of $G_{\ell}.$ Suppose the deleted edge is $(X_{i},X_{j}).$ By Lemma~\ref{Lemma: deleting an important edge}, we have $X_{i} \ncid_{G_{\ell}} X_{j} \vert \{\Pa_{G^{\prime}}(X_{j})\}$.
Hence, again because of the $\delta_{\ell}-$strong faithfulness condition, $\vert \rho_{X_{i},X_{j} \vert \Pa_{G'}(X_{j})} \vert > \delta.$ Thus, deleting this edge would worsen the score. This is a contradiction to the GES algorithm deleting this edge.

Claim 3: $C^{b}=\tilde{C}_{\ell},$ i.e., $C^{b}$ is a perfect map of $\tilde{C}_{\ell}.$

This claim follows from Lemma~\ref{Lemma: Skeleton condition} since we know from the previous claim that $C^{b}$ is an I-map of $\tilde{C}_{\ell}$ and by construction the skeletons of $C^{b}$ and $\tilde{C}_{\ell}$ are the same.

It follows from $\mathcal{C}=\tilde{\mathcal{C}}$ that $\AggregateCPDAGs(\mathcal{C})=\AggregateCPDAGs(\tilde{\mathcal{C}}).$
\end{proof}

%

\subsection{PROOF OF THEOREM~\ref{Theorem: Consistency} OF THE MAIN PAPER}

Recall that AGES combines a collection of CPDAGs obtained in the solution path of GES, where the largest CPDAG corresponds to the BIC penalty with $\lambda=\log(n)/(2n)$. In the consistency proof of GES with the BIC penalty, \cite{ChickeringGES2} used the fact that the penalized likelihood scoring criterion with the BIC penalty is locally consistent as $\log(n)/(2n) \rightarrow 0$. We note that the other penalty parameters involved in the computation of the solution path of GES do not converge to zero. This prevents us to obtain a proof of Theorem~\ref{Theorem: Consistency} of the main paper by applying the consistency result of \cite{ChickeringGES2}. A further complication is that the choices of the penalty parameters in the solution path of GES depend on the data.

In order to prove Theorem~\ref{Theorem: Consistency} of the main paper, we rely on the soundness of the oracle version of AGES (Theorem~\ref{Theorem: Equality of the CPDAGs} of the main paper). In fact, we prove consistency of AGES by showing that the solution path of GES coincides with its oracle solution path as the sample size tends to infinity. Since the number of variables is fixed and the solution path of GES depends only on the partial correlations (see Lemma~\ref{Lemma: 5.1 Preetam} and Section~\ref{Section: Computation} of the main paper), the consistency of AGES will follow from the consistency of the sample partial correlations.

\begin{proof}[Proof of Theorem~\ref{Theorem: Consistency} of the main paper]
   Given a scoring criterion, each step of GES depends on the scores of all DAGs on $p$ variables through their ranking only, where each step in the forward (backward) phase corresponds to improving the current ranking as much as possible by adding (deleting) a single edge. Let $\hat{\boldsymbol \rho}$ denote a vector consisting of the absolute values of all sample partial correlations $\hat{\rho}_{X_i,X_j|S}$, $1\leq i \leq j \leq p$ and $S \subseteq X\setminus \{X_i,X_j\}$,  in some order. It follows from Lemma~\ref{Lemma: 5.1 Preetam} that 
   the solution path of GES (for $\lambda \ge \log(n)/(2n)$) solely depends on the ranking of the elements in $\hat{\boldsymbol \gamma}$, where $\hat{\boldsymbol \gamma}$ contains the elements of $\hat{\boldsymbol \rho}$ appended with $(1-n^{-1/n})^{1/2}$, where the last element results from solving $-\log(1-\rho^2)/2 = \log(n)/(2n)$ for $\rho$.

   Similarly, an oracle solution path of GES solely depends on a ranking of the elements in $\boldsymbol \gamma$, where $\boldsymbol \gamma$ contains the elements of $\boldsymbol \rho$ appended with the value $0$, and $\boldsymbol \rho$ denotes a vector consisting of the absolute values of all partial correlations in the same order as in $\hat{\boldsymbol \rho}$. Note that there can be more than one oracle solution paths of GES depending on a rule for breaking ties. We will write $\rank\left(\hat{\boldsymbol \gamma}\right) = \rank(\boldsymbol \gamma)$ if $\rank\left(\hat{\boldsymbol \gamma}\right)$ equals a ranking of $\boldsymbol \gamma$ with some rule for breaking ties.

   Finally, we define
   \vspace{-0.05in}
   \begin{align*}
   \epsilon = \min \left\{\left| |\rho_{X_{i_1},X_{j_1}|S_1}| - |\rho_{X_{i_2},X_{j_2}|S_2}| \right| :\right. \\ \left. |\rho_{X_{i_1},X_{j_1}|S_1}| \neq  |\rho_{X_{i_2},X_{j_2}|S_2}| \right\},
   \end{align*}
   where the minimum is taken over all $1\le i_1< j_1 \le p$, $S_1\subseteq X \setminus \{X_{i_1},X_{j_1}\}$, $1\le i_2 < j_2 \le p$ and $S_2\subseteq X \setminus \{X_{i_2},X_{j_2}\}$.
   Therefore, it follows from Theorem~\ref{Theorem: Equality of the CPDAGs} of the main paper and the consistency of the sample partial correlations that
   \begin{align*}
   &~\mathbb{P}\left(\AGES(X^{(n)}) \neq A_{0} \right) \\
   \leq &~\mathbb{P}\left(\rank\left(\hat{\boldsymbol \gamma}\right) \neq \rank(\boldsymbol \gamma)\right) \\
   \leq&~ \sum_{\substack{\text{$1\leq i < j \leq p$,}\\  \\ \text{$S \subseteq X\setminus \{X_i,X_j\}$} } }\mathbb{P}\left( \left| |\hat{\rho}_{X_i,X_j|S}| - |\rho_{X_i,X_j|S}| \right| \geq \epsilon/2 \right)
   \end{align*}
   converges to zero as the sample size tends to infinity.
\end{proof}

\section{ADDITIONAL SIMULATION RESULTS WITH $p=10$}\label{Section: other empirical results}

We also ran AGES on the settings described in the main paper but with smaller sample sizes. Figures~\ref{Figure: Sample size 100} and \ref{Figure: Sample size 1000} show the results for $n=100$ and $n=1000$, respectively, based on 500 simulations per setting.

\begin{figure}[tb]
   \centering
   \includegraphics[scale=0.25]{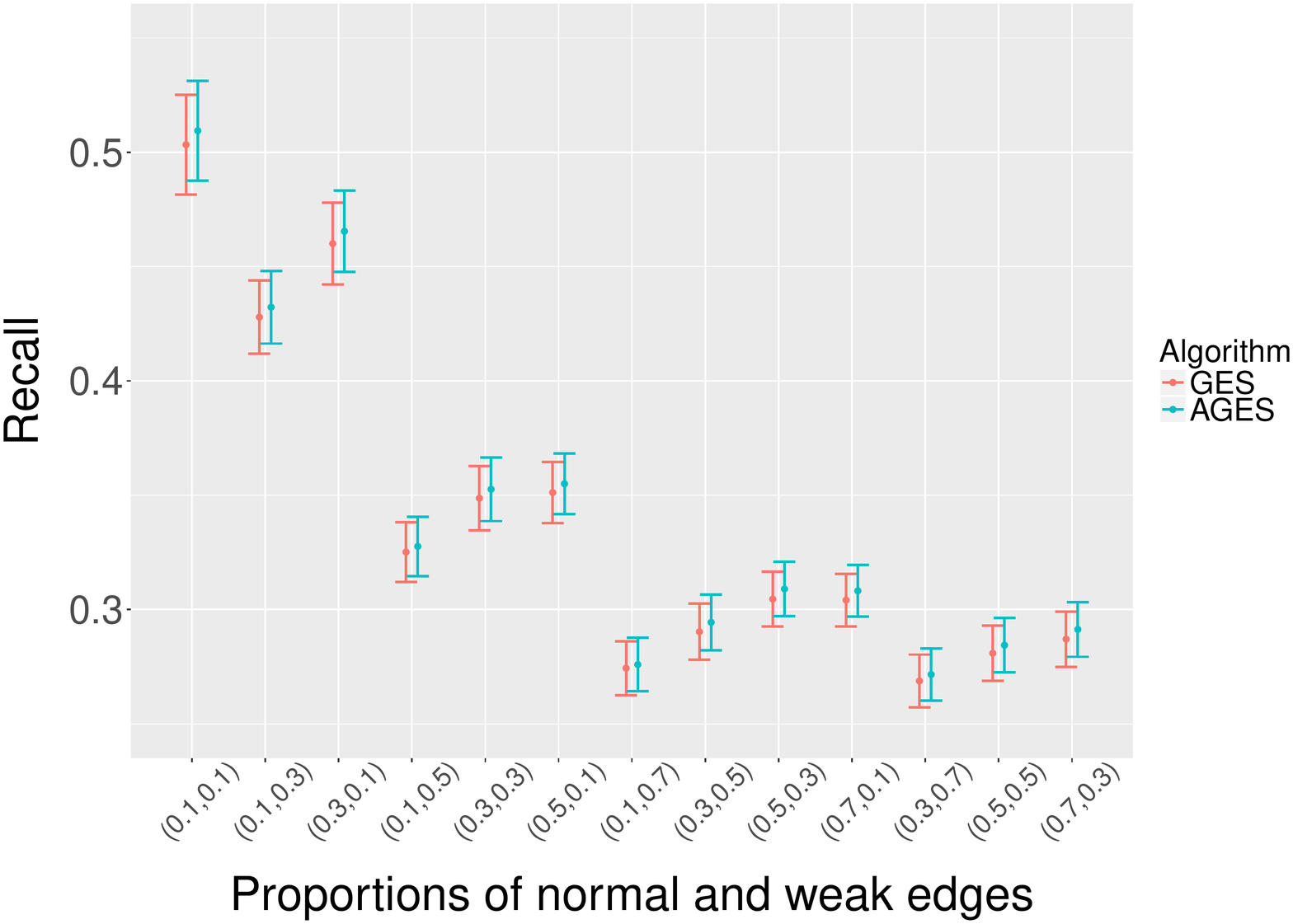}\\
   \vspace{0.5cm}
   \includegraphics[scale=0.25]{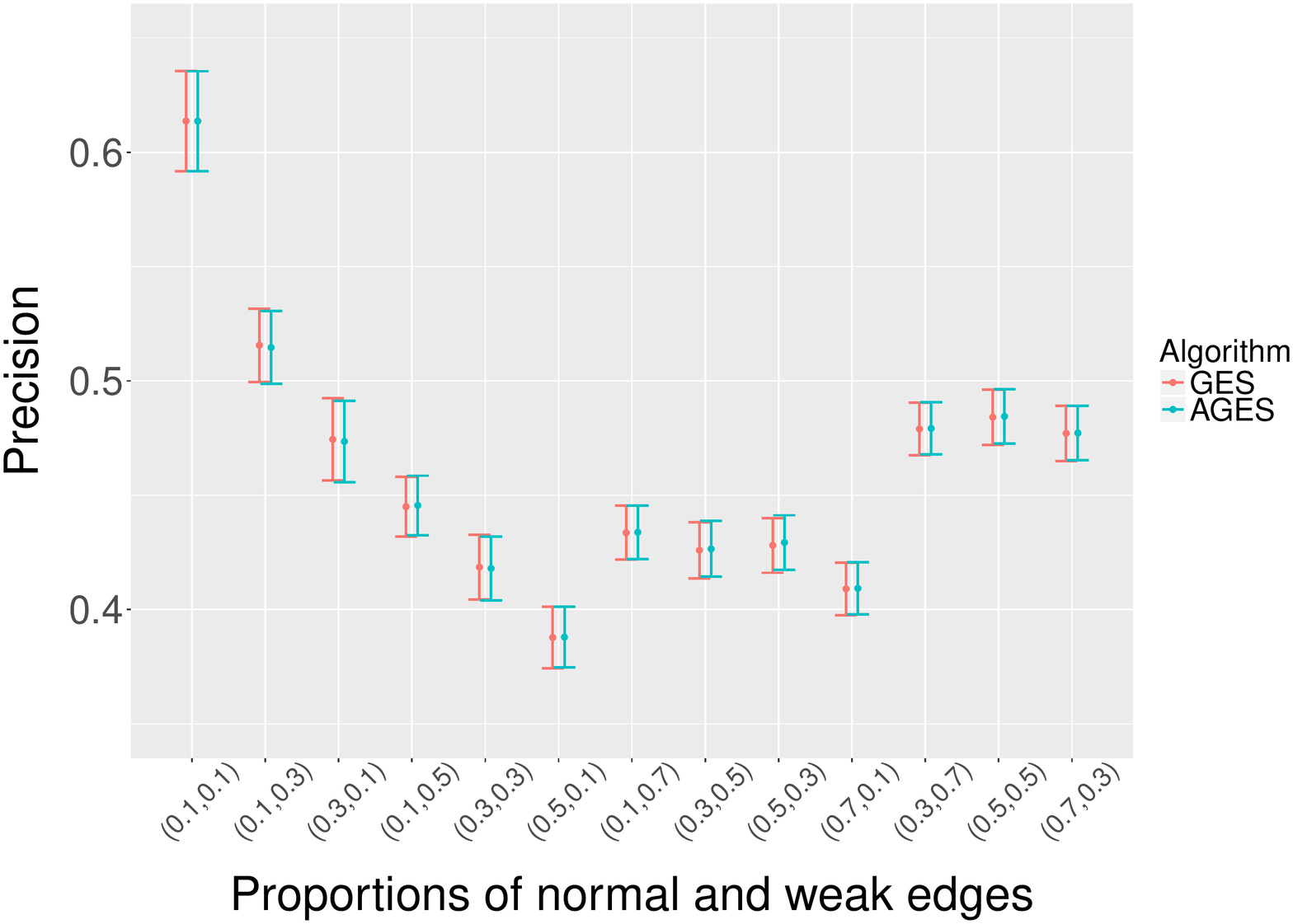}
   \caption{Mean precision and recall of GES and AGES over 500 simulations for all combinations of $(q_s,q_w) \in \{0.1, 0.3, 0.5, 0.7\}$ such that $q_s+q_w \leqslant 1$, for $p=10$, $\lambda=\log(n)/(2n)$ and $n=100$ (see Section~\ref{Section: other empirical results}). The bars in the plots correspond to $\pm$ twice the standard error of the mean.}
   \label{Figure: Sample size 100}
\end{figure}

\begin{figure}[tb]
   \centering
   \includegraphics[scale=0.25]{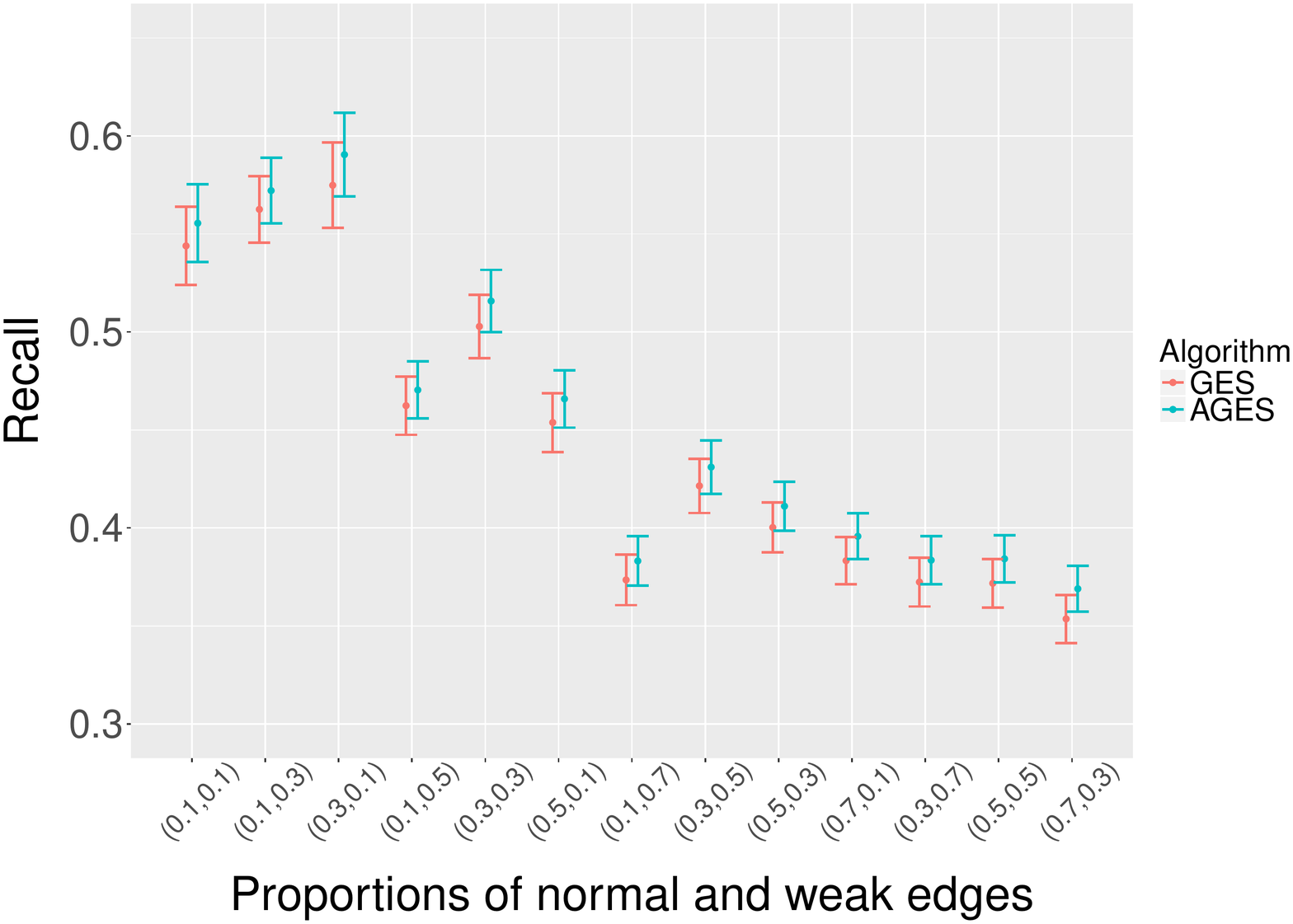}\\
   \vspace{0.5cm}
   \includegraphics[scale=0.25]{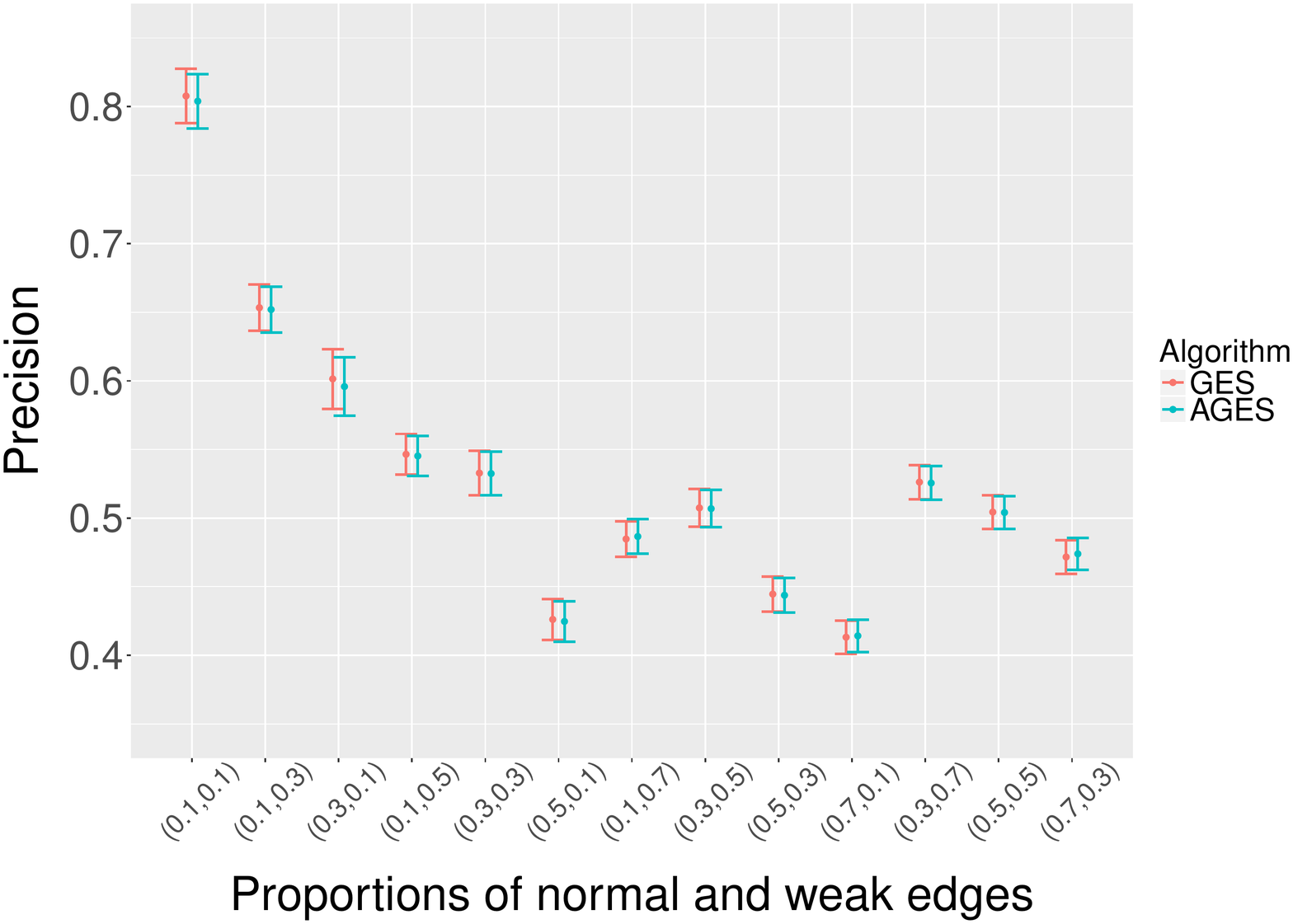}
   \caption{Mean precision and recall of GES and AGES over 500 simulations for all combinations of $(q_s,q_w) \in \{0.1, 0.3, 0.5, 0.7\}$ such that $q_s+q_w \leqslant 1$, for $p=10$, $\lambda=\log(n)/(2n)$ and $n=1000$ (see Section~\ref{Section: other empirical results}).The bars in the plots correspond to $\pm$ twice the standard error of the mean.}
   \label{Figure: Sample size 1000}
\end{figure}

For the larger sample size, $n=1000$, we see that we still gain in recall and that the precision remains roughly constant. For the smaller sample size, $n=100$, the differences become minimal. In all cases AGES performs at least as good as GES.

With a sample size of $100$ we expect to detect only partial correlations with an absolute value larger than $0.21.$ This can be derived solving $1/2\log(1-\rho^2)=-\log(n)/(2n)$ for $\rho.$ This limits the possibility of detecting weak edges, and if an edge is not contained in the output of GES it is also not contained in the output of AGES. This explains why we do not see a large improvement with smaller sample sizes. However, AGES then simply returns an APDAG which is very similar, or identical, to the CPDAG returned by GES.

\section{FURTHER SIMULATION RESULTS WITH $p=100$}
\label{Section: Simulations with p=100}

We randomly generated 500 DAGs consisting of 10 disjoint blocks of complete DAGs, where each block contains strong and weak edges with concentration probabilities $(q_s,q_w) = (0.3,0.7)$. The absolute values of the strong and weak edge weights are drawn from Unif(0.8,1.2) and Unif(0.1,0.3), respectively. The sign of each edge weight is chosen to be positive or negative with equal probabilities. The variance of the error variables are drawn from Unif(0.5,1.5).

This setting leads more often to a violation of the skeleton condition of Algorithm~\ref{Algorithm: Main(sample)} of the main paper, i.e.,\ the skeleton of the output of GES with $\lambda > \log(n)/(2n)$ is not a subset of the skeleton of the output of GES with $\lambda = \log(n)/(2n)$. This results in almost identical outputs of GES and AGES. In order to alleviate this issue, in each step of AGES with $\lambda > \log(n)/(2n)$, we replace GES with the ARGES-skeleton algorithm of \cite{Nandy}, based on the skeleton of the output of GES with $\lambda = \log(n)/(2n)$. ARGES-skeleton based on an estimated CPDAG is a hybrid algorithm that operates on a restricted search space determined by the estimated CPDAG and an adaptive modification. The adaptive modification was proposed to retain the soundness and the consistency of GES and it can be easily checked that our soundness and consistency results continue to hold if we replace GES by ARGES-skeleton in each step of AGES with $\lambda > \log(n)/(2n)$. An additional advantage of using ARGES-skeleton is that it leads to a substantial improvement in the runtime of AGES.

\begin{figure}[tb]
\centering
\includegraphics[scale=0.25]{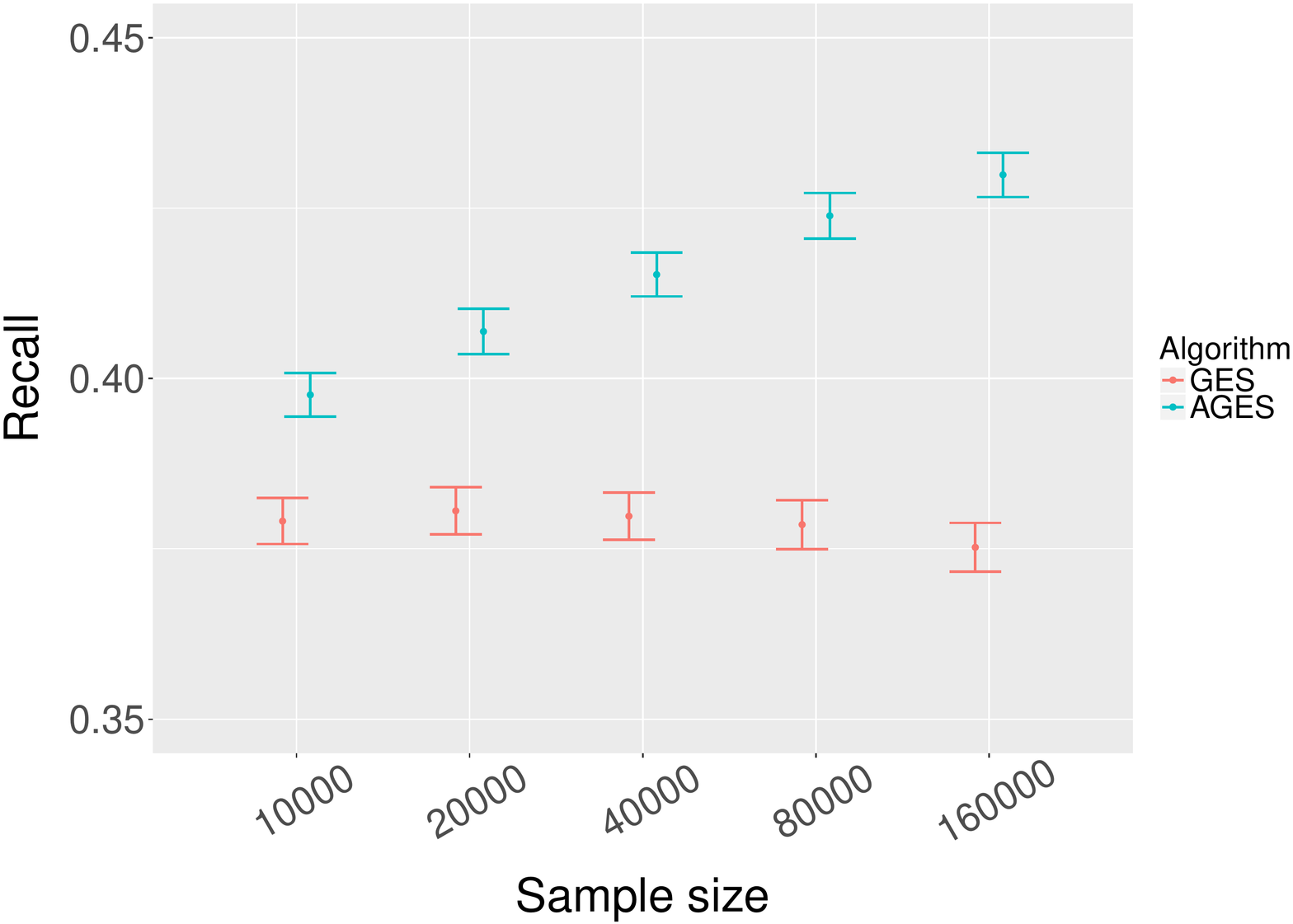}\\
\vspace{0.5cm}
\includegraphics[scale=0.25]{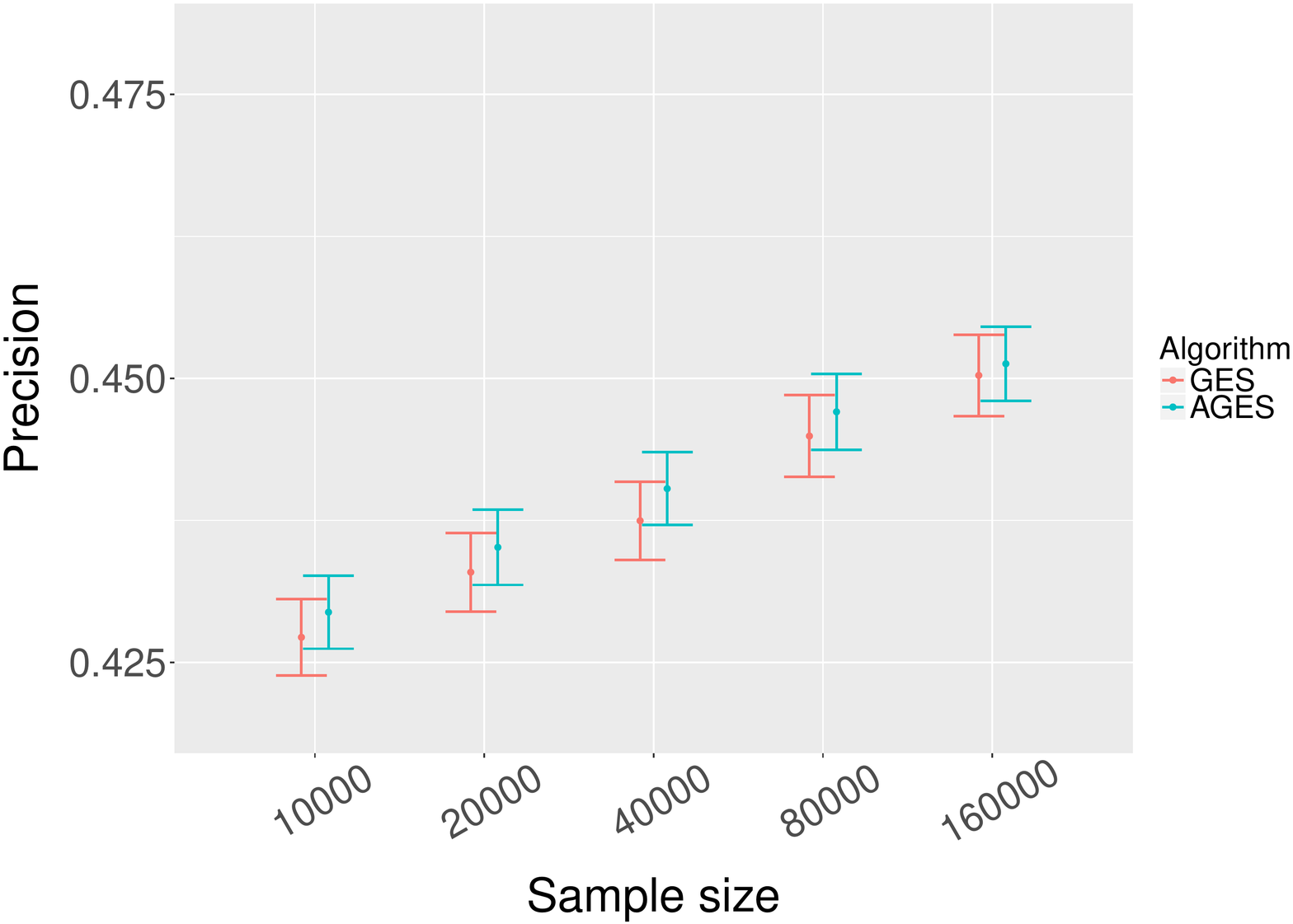}
\caption{Mean precision and recall of GES and AGES with ARGES-skeleton over 500 simulations for $(q_s,q_w) = (0.3,0.7)$, $p=100$, $\lambda=\log(n)/(2n)$, and varying sample sizes (see Section~\ref{Section: Simulations with p=100}). The bars in the plots correspond to $\pm$ twice the standard error of the mean.}
\label{fig: large p BIC}
\end{figure}

\begin{figure}[tb]
\centering
\includegraphics[scale=0.25]{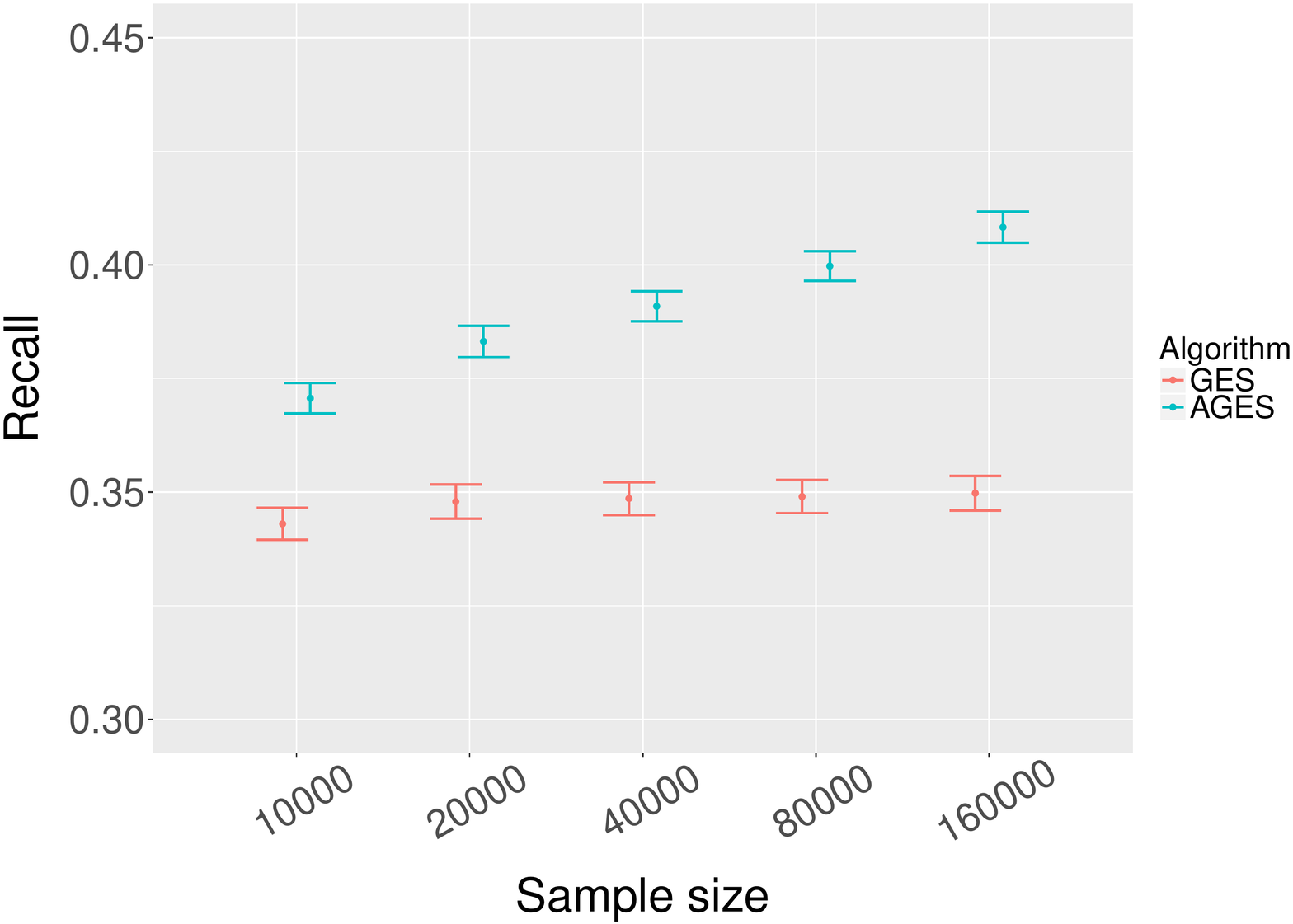}\\
\vspace{0.5cm}
\includegraphics[scale=0.25]{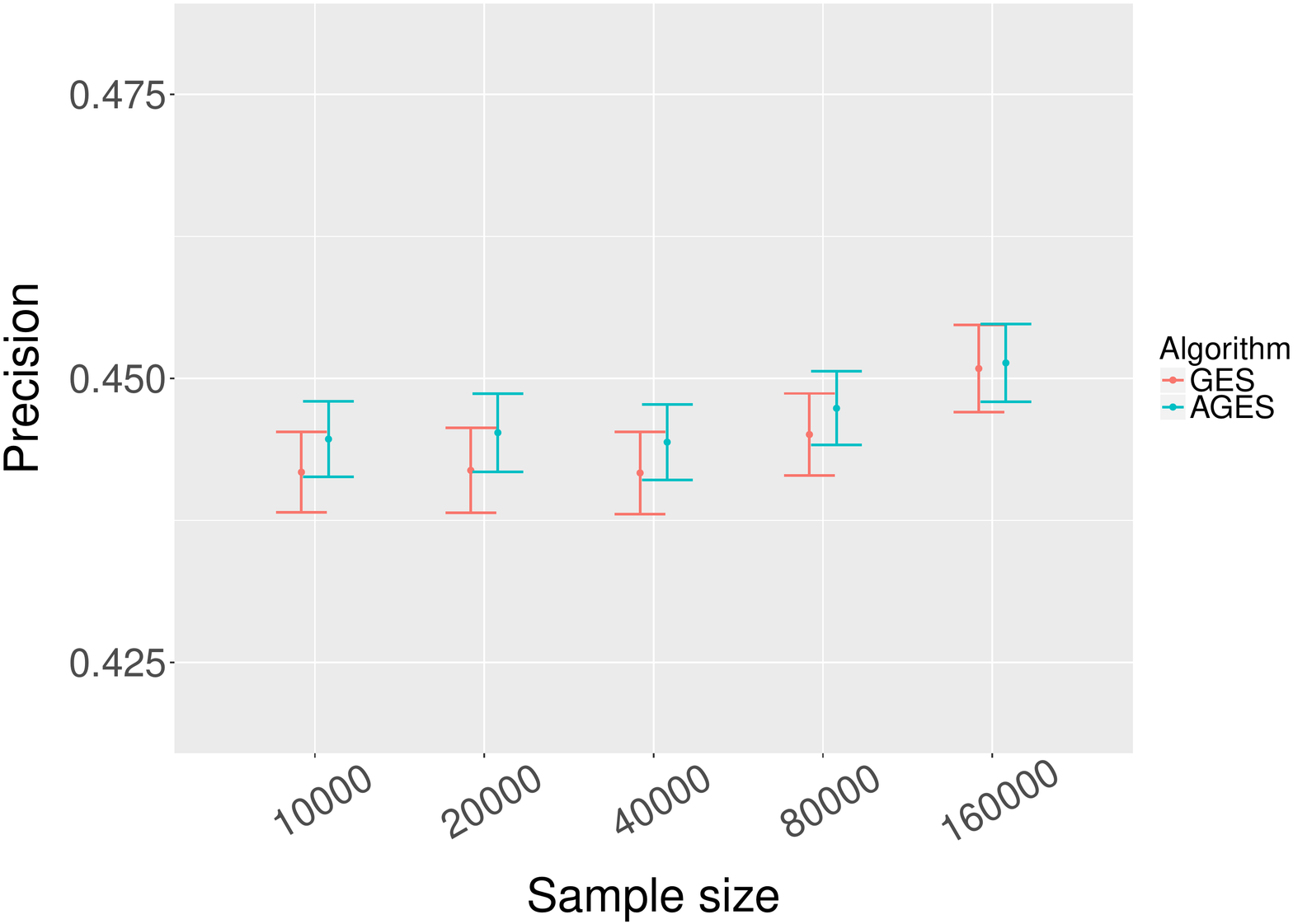}
\caption{Mean precision and recall of GES and AGES with ARGES-skeleton over 500 simulations for $(q_s,q_w) = (0.3,0.7)$, $p=100$, $\lambda=\log(n)/(2n) + \log(p)$, and varying sample sizes (see Section~\ref{Section: Simulations with p=100}). The bars in the plots correspond to $\pm$ twice the standard error of the mean.}
\label{fig: large p EBIC}
\end{figure}

Figure~\ref{fig: large p BIC} shows that AGES (based on ARGES-skeleton) achieves higher recall than GES for estimating the true directions while retaining a similar precision as GES. Unsurprisingly, the difference in the recalls of AGES and GES becomes more prominent for larger sample sizes. We obtain a similar relative performance by using the extended BIC penalty $\lambda = \log(n)/(2n) + \log(p)$ \citep[e.g.,][]{FoygelDrton2010} instead of the BIC penalty (Figure~\ref{fig: large p EBIC}).



\section{PATH STRONG FAITHFULNESS}\label{Section: path strong faithful}
To produce Figure~\ref{Figure dependence on edge weights} of the main paper we started by determining the possible APDAGs $A_0$ for each choice of the edge weights. This is done by considering the four steps in Section~\ref{Section: New target} of the main paper. In Step~\ref{Step 1} we can obtain many CPDAGs (3 with one edge, 6 with two edges, and 1 with three edges). However, once we proceed to Step~\ref{Step 2}, we note that only one DAG contains a v-structure. Hence, the orientations in the CPDAGs in Step~\ref{Step 3} are limited to this v-structure. Therefore, the only two possible APDAGs are given in Figures~\ref{Subfigure Informative APDAG} and \ref{Subfigure Uninformative APDAG} of the main paper.

Now we consider possible outputs of the oracle version of AGES for every choice of the edge weights. To compute them we have to compute all marginal and partial correlations. Then, we select the in absolute value largest marginal correlation. This corresponds to the first edge addition. Now, we consider the four remaining marginal and partial correlations between non-adjacent vertices. If the in absolute value largest partial correlation is actually a marginal correlation, then we do not obtain a v-structure and the output of AGES is Figure~\ref{Subfigure Uninformative APDAG} of the main paper. Otherwise, AGES recovers an APDAG with a v-structure.

With these results, we can compute the different areas depicted in Figure~\ref{Figure dependence on edge weights} of the main paper.

\section{AGES $\delta$-STRONG FAITHFULNESS}\label{Subsection: Strong Faithfulness}

The path strong faithfulness assumption in Theorem~\ref{Theorem: Equality of the CPDAGs} of the main paper is sufficient but not necessary for the theorem.

We now present an alternative strong faithfulness assumption which is weaker than path strong faithfulness. This new assumption is necessary and sufficient for Theorem~\ref{Theorem: Equality of the CPDAGs} of the main paper.

In a CPDAG $C=(V,E)$ we say that $S \subseteq V\setminus \{X_i\}$ is a \emph{possible parent set} of $X_i$ in $C$ if there is a DAG $G$ in the Markov equivalence class represented by $C$ such that $\Pa_G(X_i)=S$.

\begin{definition}\label{def: AGES strong faithful}
  A multivariate Gaussian distribution is said to be \emph{AGES $\delta$-strong faithful} with respect to a DAG $G$ if it holds that $X_{i} \ncid_{G} X_{j} \vert S \Rightarrow \vert \rho_{X_{i},X_{j}\vert S} \vert > \delta$ for every triple $(X_i,X_j,S)$ belonging to at least one of the following two sets:
  \begin{enumerate}
  \item Consider the output of the forward phase of oracle GES with penalty parameter $\lambda=-1/2\log(1-\delta^2)$. The first set consists of all triples $(X_{i},X_{j}, S)$ such that, in this forward phase output, $S$ is a possible parent set of $X_j$, $X_i$ is a non-descendant of $X_j$ in the DAG used to define $S$, and $X_i$ and $X_j$ are not adjacent.
  \item Consider the backward phase of oracle GES when ran with penalty parameter $\lambda$ and starting from the output of the forward phase. The second set consists of all triples $(X_{i},X_{j}, S)$ such that the edge between $X_i$ and $X_j$ has been deleted during the backward phase using $S$ as conditioning set.
  \end{enumerate}
\end{definition}

The need for a different condition becomes clear when we think about how GES operates.
In Example~\ref{Example: GES-SF}, we show why path strong faithfulness is too strong.

\begin{example}\label{Example: GES-SF}
   Consider the distribution $f$ generated from the weighted DAG $G_0$ in Figure \ref{Subfigure: GES SF true DAG} with $\varepsilon\sim N(0,I)$. The solution path of oracle GES is shown in Figures \ref{Subfigure: GES SF 4 edges}-\ref{Subfigure: GES SF 1 edge}. Note that all sub-CPDAGs found by oracle GES coincide with the CPDAGs constructed as described in Step~\ref{Step 3} of the main paper, i.e., $C_i = \tilde{C}_i$ for $0\leqslant i \leqslant 3$.

   Intuitively, we would like a condition that is satisfied if and only if the two CPDAGs coincide. However, this is not necessarily the case for the path strong faithfulness condition.

   For the CPDAG in Figure~\ref{Subfigure: GES SF 1 edge}, path strong faithfulness imposes $\delta_3$-strong faithfulness with respect to $C_3$, i.e., $\vert \rho_{X_3,X_4\vert S} \vert > \delta_3$ for all sets $S$ not containing $X_3$ or $X_4$. However, the forward phase of GES only checks the marginal correlation between $X_3$ and $X_4$. The same is true for the backward phase.


   Consider now Figure~\ref{Subfigure: GES SF 2 edges}. Path strong faithfulness imposes $\delta_2$-strong faithfulness with respect to $C_2$. For instance, it requires that $\vert \rho_{X_2,X_4 \vert X_1} \vert > \delta_2$. However, this partial correlation does not correspond to a possible edge addition. Hence, this constraint is not needed, and it is not imposed by AGES $\delta$-strong faithfulness.
   
   In this example, $\vert \rho_{X_2,X_4 \vert X_1} \vert < \delta_2 = \vert \rho_{X_1,X_3\vert X_2} \vert$. Hence, $f$ does not satisfy $\delta_2$-strong faithfulness with respect to $C_2$, but it does satisfy AGES $\delta_2$-strong faithfulness with respect to $C_2$. 
\end{example}

\begin{figure}[tb]
\centering
\subfloat[True DAG $G_{0}$.]{
\label{Subfigure: GES SF true DAG}
\begin{tikzpicture}[scale=0.7]
\node[node] (X) at (-1.2,1) {$X_1$};
\node[node] (Y) at (1.2,1) {$X_2$};
\node[node] (Z) at (1.2,-1) {$X_3$};
\node[node] (W) at (-1.2,-1) {$X_4$};
\draw[->] (X) to [above] node{$0.1$} (Y);
\draw[->] (X) to [above right]  node{$1$} (Z);
\draw[->] (Y) to [right] node{$1$} (Z);
\draw[->] (Z) to [above] node{$1$} (W);
\end{tikzpicture}
}
\hfill
\subfloat[CPDAG $C_{0}=\GES_{\lambda_{0}}(f).$]{
\label{Subfigure: GES SF 4 edges}
\begin{tikzpicture}[scale=0.7]
\node[node] (X) at (-1.2,1) {$X_1$};
\node[node] (Y) at (1.2,1) {$X_2$};
\node[node] (Z) at (1.2,-1) {$X_3$};
\node[node] (W) at (-1.2,-1) {$X_4$};
\draw[-] (X) to (Y);
\draw[-] (X) to (Z);
\draw[-] (Y) to (Z);
\draw[-] (Z) to (W);
\end{tikzpicture}
}
\hfill
\subfloat[CPDAG $C_{1}=\GES_{\lambda_{1}}(f)$.]{
\label{Subfigure: GES SF 3 edges}
\begin{tikzpicture}[scale=0.7]
\node[node] (X) at (-1.2,1) {$X_1$};
\node[node] (Y) at (1.2,1) {$X_2$};
\node[node] (Z) at (1.2,-1) {$X_3$};
\node[node] (W) at (-1.2,-1) {$X_4$};
\draw[->] (X) to (Z);
\draw[->] (Y) to (Z);
\draw[->] (Z) to (W);
\end{tikzpicture}
}
\vfill
\hspace{0.2cm}
\hfill
\subfloat[CPDAG $C_{2}=\GES_{\lambda_{2}}(f)$.]{
\label{Subfigure: GES SF 2 edges}
\begin{tikzpicture}[scale=0.7]
\node[node] (X) at (-1.2,1) {$X_1$};
\node[node] (Y) at (1.2,1) {$X_2$};
\node[node] (Z) at (1.2,-1) {$X_3$};
\node[node] (W) at (-1.2,-1) {$X_4$};
\draw[-] (Y) to (Z);
\draw[-] (Z) to (W);
\end{tikzpicture}
}
\hfill
\subfloat[CPDAG $C_{3}=\GES_{\lambda_{3}}(f)$.]{
\label{Subfigure: GES SF 1 edge}
\begin{tikzpicture}[scale=0.7]
\node[node] (X) at (-1.2,1) {$X_1$};
\node[node] (Y) at (1.2,1) {$X_2$};
\node[node] (Z) at (1.2,-1) {$X_3$};
\node[node] (W) at (-1.2,-1) {$X_4$};
\draw[-] (Z) to (W);
\end{tikzpicture}
}
\hfill
\hspace{0.2cm}
\caption{Graphs corresponding to Example~\ref{Example: GES-SF}. Figure~\ref{Subfigure: GES SF true DAG} shows the true underlying DAG $G_0$. Figure~\ref{Subfigure: GES SF 4 edges} - \ref{Subfigure: GES SF 1 edge} show the sub-CPDAG oracle GES found. }
\label{Figure: GES-delta-SF}
\end{figure}

The following lemma states that the AGES $\delta$-strong faithfulness assumption is necessary and sufficient for Claim~1 and Claim~2 in the proof of Theorem~\ref{Theorem: Equality of the CPDAGs} of the main paper.

\begin{lemma}\label{Lemma GES SF}
Given a multivariate Gaussian distribution $f$ and a CPDAG $C$ on the same set of vertices, $\GES(f,\lambda)$ with $\lambda=-1/2\log(1-\delta^2)$ is an I-map of $C$ if and only if $f$ is AGES $\delta$-strong faithful with respect to $C$.
\end{lemma}



\begin{proof}
For a CPDAG $C$, we use the notation $X_i \cid_C X_j \vert S$ to denote that $X_i \cid_G X_j \vert S$ in any DAG $G$ in the Markov equivalence class described by $C$.

We first prove the ``if" part. Thus, assume that $f$ is AGES $\delta$-strong faithful with respect to $C$. We consider running oracle GES with $\lambda=-1/2 \log(1-\delta^2)$, and denote by $C^{f}$ and $C^{b}$ the output of the forward and backward phase, respectively.

\textit{Claim 1:} $C^{f}$ is an I-map of $C$, i.e., all d-separation constraints true in $C^{f}$ are also true in $C$.

\textit{Proof of Claim 1:}

For each triple $(X_i,X_j,S)$ contained in the first set of Definition~\ref{def: AGES strong faithful}, we have $\vert \rho_{X_i,X_j \vert S} \vert < \delta$, since otherwise there would have been another edge addition. From AGES $\delta$-strong faithfulness, it follows that $X_i \cid_C X_j \vert S$. Since this set of triples characterizes the d-separations that hold in $C^f$, all d-separations that hold in $C^f$ also hold in $C$.

\textit{Claim 2:} $C^{b}$ is an I-map of $C$, i.e., all d-separation constraints true in $C^{b}$ are also true in $C$.

\textit{Proof of Claim 2:}
By Claim 1 the backward phase starts with an I-map of $C$. Suppose it ends with a CPDAG that is not an I-map of $C$. Then, at some point there is an edge deletion which turns a DAG $G$ that is an I-map of $C$ into a DAG $G^{\prime}$ that is no longer an I-map of $C$. Suppose the deleted edge is $(X_{i},X_{j})$. By Lemma~\ref{Lemma: deleting an important edge}, we have $X_{i} \ncid_{C} X_{j} \vert \Pa_{G^{\prime}}(X_{j})$.
Since the edge has been deleted, the corresponding triple $(X_i,X_j,\Pa_{G^{\prime}}(X_{j}))$ is contained in the second set of Definition~\ref{def: AGES strong faithful}. Hence, by AGES $\delta$-strong faithfulness, we obtain $\vert \rho_{X_{i},X_{j} \vert \Pa_{G'}(X_{j})} \vert > \delta.$ Thus, deleting this edge would worsen the score. This is a contradiction to the GES algorithm deleting this edge.

We now prove the ``only if'' part. Thus, suppose there is a triple $(X_i,X_j,S)$ in one of the sets in Definition~\ref{def: AGES strong faithful} such that $\vert \rho_{X_{i},X_{j} \vert S} \vert < \delta$ and $X_{i} \ncid_{C} X_{j} \vert S$. 

Suppose first that this triple concerns the first set. Since all triples in the first set characterize the d-separations that hold in $C^f$, we know that $X_i \cid_{C^f} X_j \vert S$. Therefore, $C^f$ is not an I-map of $C$. Hence, $C^b$ is certainly not an I-map of $C$.

Next, suppose the triple concerns the second set. This means that at some point there is an edge deletion which turns a DAG $G$ into a DAG $G^{\prime}$ by deleting the edge $X_i \to X_j$, using $S$ as conditioning set. This means that $S = \Pa_G(X_j) \setminus\{X_i\} = \Pa_{G'}(X_j)$. In the resulting DAG $G'$, $X_i$ and $X_j$ are therefore d-separated given $S$. But we know that $X_i \ncid_{C} X_j \vert S$. Hence, $C^b$ is not an I-map of $C$.
\end{proof}

We analysed how often the AGES $\delta$-strong faithfulness assumption is met in the simulations presented in the main paper, as well as how often oracle AGES is able to find the correct APDAG. Lemma~\ref{Lemma GES SF} provides a necessary and sufficient condition for the equality of the CPDAGs of Theorem~\ref{Theorem: Equality of the CPDAGs} of the main paper. For the equality of the APDAGs this condition is only sufficient.

\begin{figure}[tb]

\centering
\includegraphics[trim={32cm 0 32cm 0},scale=0.55]{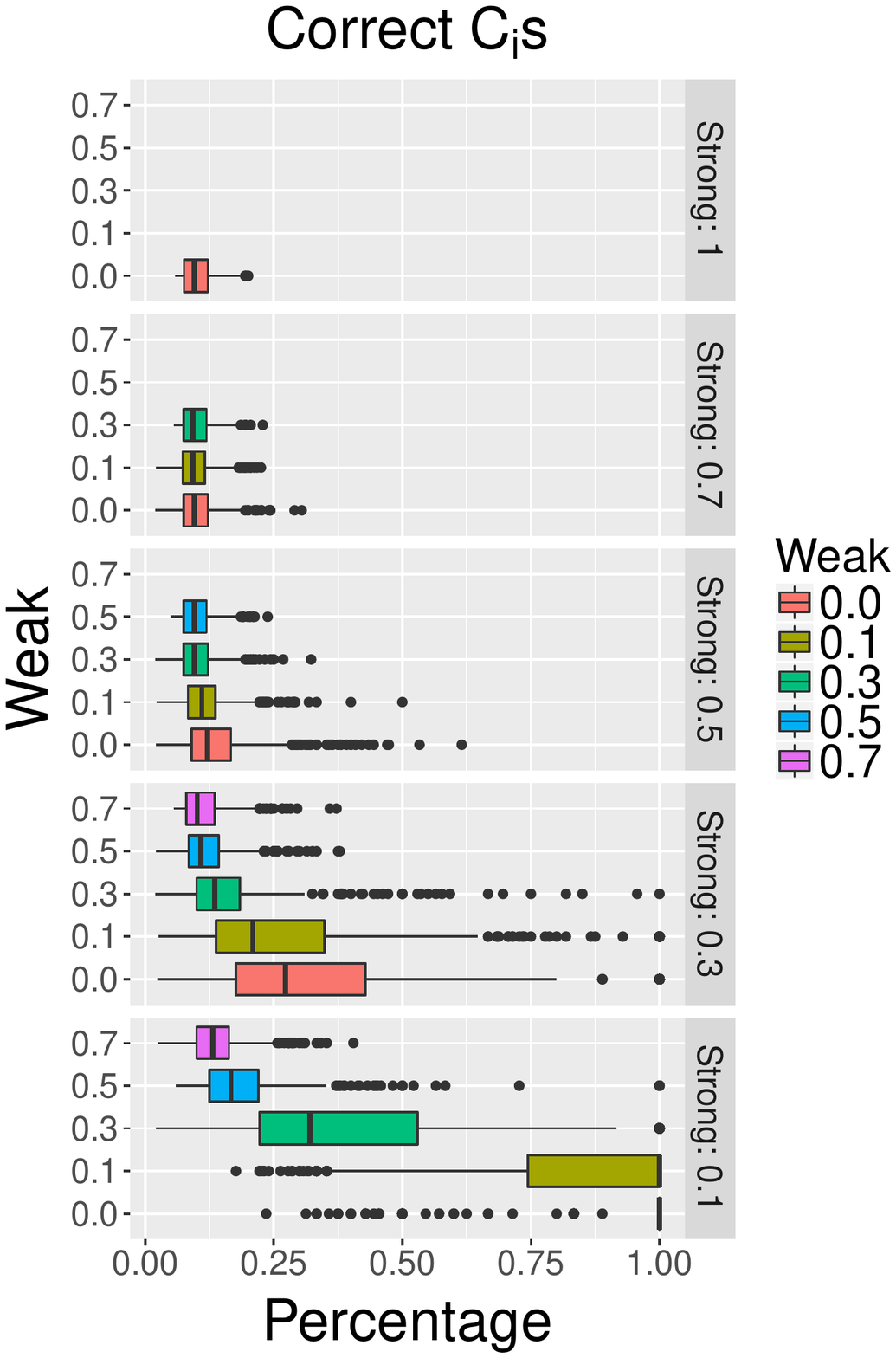}

\caption{Boxplots of the proportion of correct sub-CPDAGs $\tilde{C}_1, \ldots, \tilde{C}_k$ (as defined in Step~\ref{Step 3} of Section~\ref{Section: New target} of the main paper) found by oracle AGES in each solution path (see Section \ref{Subsection: Strong Faithfulness}). The different colors represent the different proportions of weak edges. The plots are grouped by the proportion of strong edges.}
\label{Figure: Boxplots Correct C_is}
\end{figure}

\begin{figure}[tb]

\centering
\includegraphics[trim={32cm 0 32cm 0},scale=0.55]{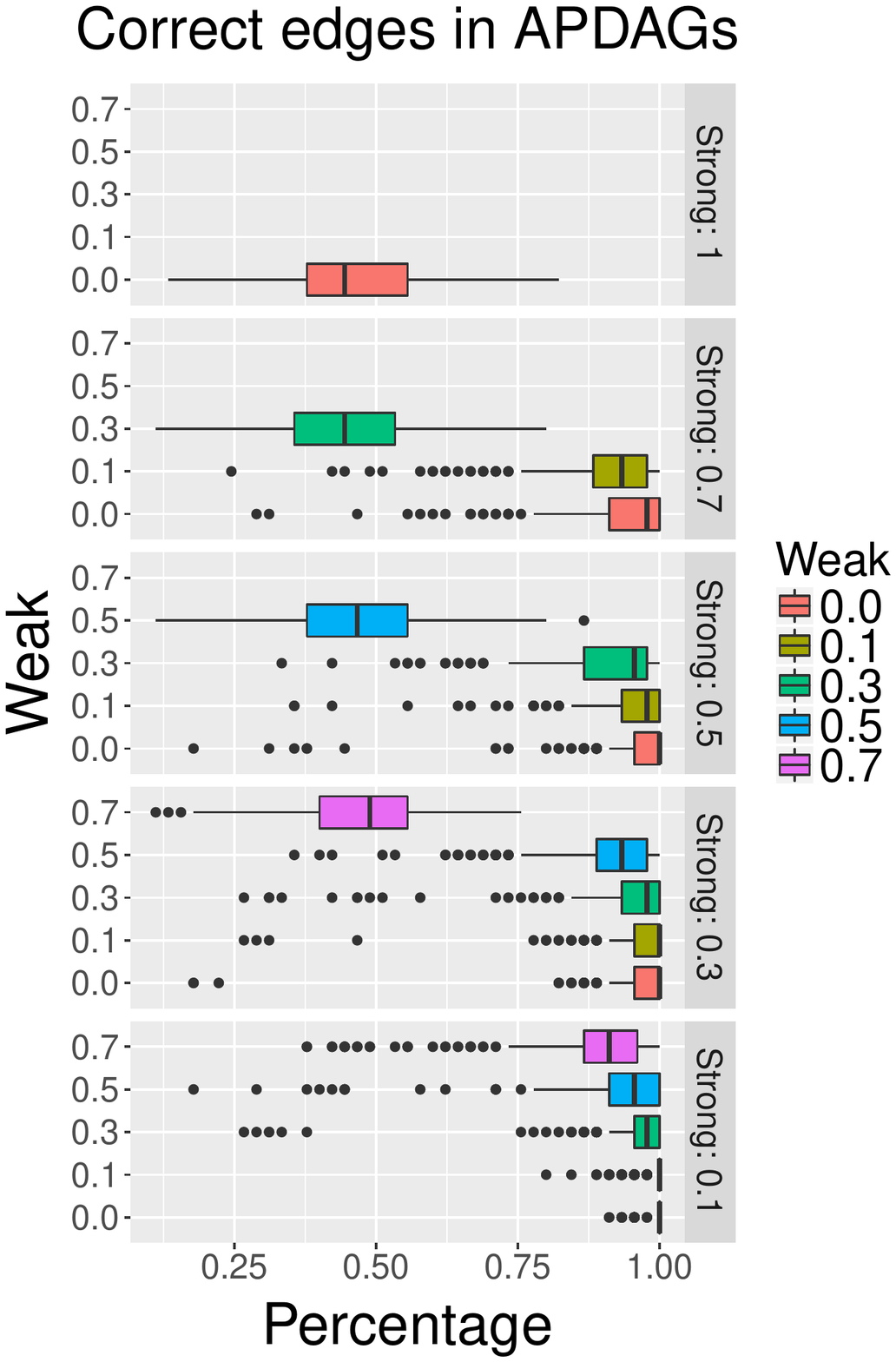}

\caption{Boxplots of the proportion of edge orientations in the APDAGs found by oracle AGES that are equal to the edge orientations in the true APDAGs (see Section \ref{Subsection: Strong Faithfulness}). The different colors represent the different proportions of weak edges. The plots are grouped by the proportion of strong edges.}
\label{Figure: Boxplots APDAGs}
\end{figure}

Figure~\ref{Figure: Boxplots Correct C_is} shows the proportion of correct sub-CPDAGs $\tilde{C}_1, \ldots, \tilde{C}_k$ (as defined in Step~\ref{Step 3} of Section~\ref{Section: New target} of the main paper) found by oracle AGES in each solution path and for all simulated settings. We can see that the sparsity of the true underlying DAG plays an important role in the satisfiability of the assumption. We can also see that for the same total sparsity, the settings with more weak edges produce better results.

Even though the AGES $\delta$-strong faithfulness assumption is not very often satisfied for denser graphs, it is much weaker than the classical $\delta$-strong faithfulness assumption. Indeed, we verified that the $\delta$-strong faithfulness assumption is rarely satisfied even for single sub-CPDAGs $C_i$.

Figure~\ref{Figure: Boxplots APDAGs} shows the proportion of edge orientations in the APDAGs found by oracle AGES that are equal to the edge orientations in the true APDAGs. With equal edge orientations, we mean that the edges have to be exactly equal. For example, an edge that is oriented in the APDAG found by oracle AGES, but oriented the other way around or unoriented in the true APDAG counts as an error. We see that in many settings AGES can correctly find a large proportion of the edge orientations. 


\section{APPLICATION TO DATA FROM SACHS ET AL., 2005}\label{Section: Real data suppmat}

We log-transformed the data because they were heavily right skewed. 
Based on the network provided in Figure~2 of \cite{Sachs}, we produced the DAG depicted in Figure~\ref{Figure: PDAG real data} that we used as partial ground truth. In the presented network, only two variables are connected by a bi-directed edge, meaning that there is a feedback loop between them. To be more conservative, we omitted this edge.

For the comparison of GES and AGES we need to account for the interventions done in the 14 experimental conditions. Following \cite{Mooij}, we distinguish between an intervention that changes the abundance of a molecule and an intervention that changes the activity of a molecule. Interventions that change the abundance of a molecule can be treated as do-interventions \citep{Pearl2009}, i.e., we delete the edges between the variable and its parents. Activity interventions, however, change the relationship with the children, but the causal connection remains. For this reason, we do not delete edges for such interventions. We also do not distinguish between an activation and an inhibition of a molecule. All this information is provided in Table 1 of \cite{Sachs}.

The only abundance intervention done in the six experimental conditions we consider in Table 1 of the main paper is experimental condition 5. This intervention concerns $PIP2$. For this reason, when comparing the outputs of GES and AGES we need to consider the DAG in Figure~\ref{Figure: PDAG real data} with the edge $PLC_{\gamma} \rightarrow PIP2$ deleted. For the other five experimental conditions we used the DAG depicted in Figure~\ref{Figure: PDAG real data} as ground truth.
\begin{figure}[tb]
\begin{center}
\begin{tikzpicture}[scale=1.2]
\node[node] (X1) at (0,3) {$PLC_{\gamma}$};
\node[node] (X2) at (2,3) {$PIP3$};
\node[node] (X3) at (-2,3) {$PIP2$};
\node[node] (X4) at (-2,1) {$PKC$};
\node[node] (X5) at (2,1) {$Akt$};
\node[node] (X6) at (2,0) {$PKA$};
\node[node] (X7) at (0,1) {$JNK$};
\node[node] (X8) at (-1,0) {$Raf$};
\node[node] (X9) at (1,2) {$p38$};
\node[node] (X10) at (-2,-1) {$Mek1/2$};
\node[node] (X11) at (2,-1) {$Erk1/2$};
\draw[thick,->] (X1) to  (X3);
\draw[thick,->] (X1) to  (X4);
\draw[thick,->] (X2) to  (X1);
\draw[thick,->] (X2) to  (X5);
\draw[thick,->] (X3) to  (X4);
\draw[thick,->] (X4) to  (X8);
\draw[thick,->] (X4) to  (X9);
\draw[thick,->] (X4) to  (X7);
\draw[thick,->] (X4) to  (X10);
\draw[thick,->] (X6) to  (X7);
\draw[thick,->] (X6) to  (X8);
\draw[thick,->] (X6) to  (X9);
\draw[thick,->] (X6) to  (X11);
\draw[thick,->] (X6) to  (X5);
\draw[thick,->] (X8) to  (X10);
\draw[thick,->] (X10) to  (X11);
\end{tikzpicture}
\end{center}
\caption{The DAG used as partial ground truth derived from the conventionally accepted network \citep{Sachs}.}
\label{Figure: PDAG real data}
\end{figure}
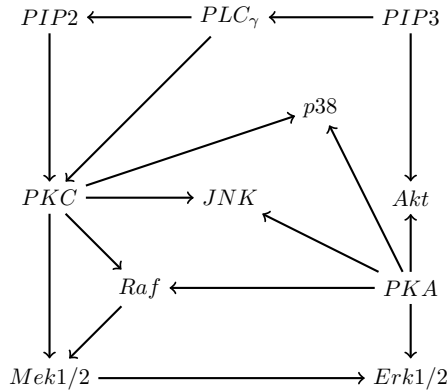

\end{document}